\documentclass{article}
\usepackage[margin=1in]{geometry}

\usepackage{subcaption}

\usepackage[utf8]{inputenc} % allow utf-8 input
\usepackage[T1]{fontenc}    % use 8-bit T1 fonts
\usepackage[hidelinks,colorlinks=true,citecolor=cyan]{hyperref}       % hyperlinks
\usepackage{url}            % simple URL typesetting
\usepackage{booktabs}       % professional-quality tables
\usepackage{amsfonts}       % blackboard math symbols
\usepackage{nicefrac}       % compact symbols for 1/2, etc.
\usepackage{microtype}      % microtypography
\usepackage{xcolor}         % colors

\usepackage{amssymb}
\usepackage{amsmath}
\usepackage{amsthm}
\usepackage{mathtools}

\usepackage{multirow}
\usepackage{multicol}

\usepackage{thmtools, thm-restate}
\usepackage{comment}
\usepackage{breqn}
\usepackage{braket}
\usepackage{ifthen}
\usepackage{tikz}
\usetikzlibrary{calc} 
\usetikzlibrary{arrows.meta}
\usetikzlibrary{positioning}

\usepackage{cleveref}

\usepackage{natbib}

\makeatletter
\AddToHook{cmd/appendix/before}{\def\cref@section@alias{appendix}\def\cref@subsection@alias{appendix}}
\makeatother

\usepackage{algorithm}
\usepackage[noend]{algpseudocode}
\algrenewcommand\algorithmicrequire{\textbf{Input:}}
\algrenewcommand\algorithmicensure{\textbf{Output:}}
\newcommand{\algrule}[1][.2pt]{\par\vskip.2\baselineskip\hrule height #1\par\vskip.2\baselineskip}

\newtheorem{theorem}{Theorem}
\newtheorem{lemma}{Lemma}

\newtheorem{definition}{Definition}

\newcommand{\R}{\mathbb{R}}
\newcommand{\E}{\mathbb{E}}

\newcommand{\poly}{\mathrm{poly}}

\newcommand{\norm}[1]{\|#1\|}

\newcommand{\ip}[2]{\left\langle #1, #2 \right\rangle}

\newenvironment{reptheorem}[1]{%
  \begingroup
  \renewcommand{\thetheorem}{\ref{#1}}%
  \begin{theorem}
}{%
  \end{theorem}
  \endgroup
}

\newenvironment{replemma}[1]{%
  \begingroup
  \renewcommand{\thelemma}{\ref{#1}}%
  \begin{lemma}
}{%
  \end{lemma}
  \endgroup
}

\title{On Worst-Case Guarantees for Graph-Based Vector Search}

\author{%
Majid Daliri\thanks{New York University.
  \href{mailto:daliri.majid@nyu.edu}{daliri.majid@nyu.edu},
  \href{mailto:cmusco@nyu.edu}{cmusco@nyu.edu},
  \href{mailto:torsten.suel@nyu.edu}{torsten.suel@nyu.edu}}
\and Mohammadreza Daneshvaramoli\thanks{UMass Amherst.
  \href{mailto:mdaneshvaram@cs.umass.edu}{mdaneshvaram@cs.umass.edu},
  \href{mailto:purnadutta@umass.edu}{purnadutta@umass.edu},
  \href{mailto:cmusco@cs.umass.edu}{cmusco@cs.umass.edu}}
\and Purna Dutta\footnotemark[2]
\and Edo Liberty\thanks{Google.
  \href{mailto:edo.liberty@gmail.com}{edo.liberty@gmail.com}}
\and Cameron Musco\footnotemark[2]
\and Christopher Musco\footnotemark[1]
\and Torsten Suel\footnotemark[1]
}

  \usepackage{nth}
  \usepackage{intcalc}

  \newcommand{\cAAAI}[1]{AAAI\ Conference\ on\ Artificial (AAAI)}

\begin{document}

\maketitle

\begin{abstract}
        Graph-based nearest neighbor search algorithms like HNSW, DiskANN, NSG, and others have surged in popularity in recent years due to their excellent empirical performance in comparison to space-partitioning methods like IVF data structures and Locality Sensitive Hashing (LSH). However, in contrast to LSH and related methods, obtaining strong worst-case theoretical guarantees for graph-based search has proven challenging.
    In this paper, we prove a strong lower bound that helps explain the lack of progress towards such guarantees. We show that no standard graph-based method (any that runs a variant of greedy search on a flat or hierarchical graph) can achieve sublinear query time for approximate nearest neighbor search in Euclidean space, even for an arbitrarily large approximation factor.
    On the other hand, we prove that sublinear query time \emph{is possible} if we allow for \emph{Steiner points} in the search graph. We obtain bounds comparable to LSH, showing e.g., that graph-based search with Steiner points can solve $c$-approximate nearest neighbor search in Euclidean space with query time $\tilde{O}(n^{1/(c-1)})$. We complement this result with initial experiments exploring how Steiner points can be incorporated into practical graph-based methods.
\end{abstract}

\section{Introduction}
High-dimensional nearest neighbor search is a fundamental, and notoriously difficult, algorithmic problem. The goal is to pre-process a vector database $X = \{x_1, \ldots, x_n\} \subset \R^d$ so that we can efficiently answer ``nearest neighbor queries''. That is, given any $q\in \R^d$, the goal is to quickly find a closest (or near-closest) vector to $q$ in $X$. Closeness is commonly measured by Euclidean distance or cosine similarity, but other metrics are also studied \citep{IM98,JayaramWaingartenZhang:2024,DhulipalaHadianJayaram:2026}. 

Nearest neighbor search has been studied for decades, becoming increasingly important as the problem now serves as a backbone for document and multimedia search systems based on machine-learned vector embeddings \citep{LewisPerezPiktus:2020,Bruch:2024,DouzeGuzhvaDeng:2026,SimhadriAumullerDouze:2026}. While it can be solved naively in $O(dn)$ time by simply comparing $q$ to every vector in $X$, in many settings, the database $X$ can involve billions or trillions of vectors with 100s to 1000s of dimensions. The goal is thus to design data structures that permit search in time \emph{sublinear} in $n$. While infeasible for exact search, sublinear-time is known to be possible if we allow for approximation \citep{Kleinberg:1997,KushilevitzOstrovskyRabani:1998}. 

The most well-known sublinear-time technique with provable guarantees is the \emph{Locality Sensitive Hashing} (LSH) framework of \cite{IM98}. Focusing on Euclidean distance for concreteness, LSH methods return, for any query $q$ and chosen approximation factor $c \geq 1$, a vector $\tilde{x} \in X$ such that, with high probability,
$\|\tilde{x} - q\|_2 \leq c \cdot \min_{x\in X} \|x - q\|_2.$
\cite{IM98} achieve this bound in $\tilde{O}(dn^{1/(c-1)})$ time using a polynomially sized data structure and polynomial preprocessing time. This was subsequently improved to $\tilde{O}(dn^{1/c})$ and later to $\tilde{O}(dn^{1/c^2 + o(1)})$ \citep{DIIM04,AndoniIndyk:2008}. Even stronger bounds can be obtained via  ``data-dependent'' LSH variants \citep{AndoniIndykNguyen:2014,AndoniRazenshteyn:2015}. Such bounds are striking: e.g., when $c = 2$, $n^{1/c^2} < 200$ even for a data set with $n = 1 $ billion points.

Moreover, in addition to strong theoretical guarantees, LSH methods are relatively simple to implement and work well in practice \citep{LvJosephsonWang:2007,AndoniIndykLaarhoven:2015,AumullerBernhardssonFaithfull:2020}. Closely related data-dependent space-partitioning methods based on clustering and product quantization have shown even better performance, forming the basis of many web-scale nearest-neighbor search systems 
\citep{JegouDouzeSchmid:2011,GuoSunLindgren:2020,JohnsonDouzeJegou:2021,SunSimchaDopson:2023}. 

\subsection{Graph-Based Nearest Neighbor Search.}
Despite the success of space-partitioning methods like LSH, a new class of \emph{graph-based nearest neighbor search} algorithms has recently emerged as a competitor. Graph-based methods include Microsoft's DiskANN \citep{KrishnaswamyManoharSimhadri:2024}, the Hierarchical Navigable Small World (HNSW) method \citep{MalkovYashunin:2020}, and many related algorithms \citep{MalkovPonomarenkoLogvinov:2014,FuCai:2016,HarwoodDrummond:2016,FuXiangWang:2019}. 
These methods build on ideas dating back to the 1990s \citep{AryaMount:1993,Clarkson:1994,Navarro:1999,Kleinberg:2000,krauthgamer2004navigating}. The basic idea is to construct a sparse search graph $G$ with a node corresponding to each point in $X$. Queries are served by running \emph{greedy search} in $G$: we start at a node $s$, and traverse the graph by moving to neighbors that reduce the distance to the query.% move to $s$'s neighbor that is closest to the query $q$, and proceed in this way until reaching a local minimum. 

Graph-based methods have frequently topped nearest neighbor search benchmarks, offering improved accuracy vs. runtime tradeoffs  \citep{AumullerBernhardssonFaithfull:2020,SimhadriWilliamsAumuller:2022,ManoharShenBlelloch:2024,SimhadriAumullerDouze:2026}. 
Naturally, these empirical results have spurred interest in theoretically explaining the  success of graph-based methods \citep{Laarhoven:2018,ProkhorenkovaShekhovtsov:2020,PengChoiChan:2023,XieYuLiu:2025,GollapudiKrishnaswamyShiragur:2025,Al-JazzaziDiwanGou:2025,Har-PeledRaichelRobson:2026}. However, all existing theoretical bounds only hold under strong dataset assumptions (uniformly random, sufficiently dense, etc.), or require an exponential dependence on the data dimension  $d$, or some other measure of intrinsic dimension like doubling dimension \citep{AryaMount:1993,Clarkson:1994,IndykXu:2023,lu2025proximity}. 

These bounds only imply sublinear query time when $d$ is a small multiple of $\log n$.
%As an example, a result of \cite{IndykXu:2023} shows that the graph-based DiskANN algorithm can support $c$-approximate nearest neighbor search in roughly  $\tilde{O}\left( (\frac{c+1}{c-1} )^d4^d\right)$ time for datasets with doubling dimension $d$. This time is only sublinear in $n$ when $d$ is a small multiple of $\log n$.
A major open question is if such an assumption is necessary to obtain performance guarantees for  graph-based methods. Concretely: 
\begin{quote}
% \vspace{-.5em}
\emph{Is there any graph-based nearest neighbor search method that, like LSH, can support search queries in $o(nd)$ time for worst-case, high-dimensional datasets?}
\end{quote}
\vspace{-1em}

\section{Our Contributions}
\label{sec:contributions}
To answer the question above, and specifically to prove any query-time lower bounds, we first need to formalize what a ``graph-based search method'' is. In this work, we consider a generic ``flat search graph'' model that captures all common methods that we are aware of (DiskANN, NSG, FANNG, EFANNA, etc.), with the exception of HNSW \citep{MalkovYashunin:2020}, which performs search on a hierarchy of graphs. We will discuss separately how to extend our results to HNSW and other hierarchical methods later.

Details on our flat search graph model are included in \Cref{sec:setup}. To summarize, we consider methods that pre-processes the dataset, $X$, to produce a directed search graph, $G$. Upon receiving a query, $q\in \R^d$, the method runs \emph{generalized greedy search} on $G$. Specifically, it maintains a list of visited nodes, initially populated with a starting node $s$. At each step, it pops the node off the list that is closest to $q$, and adds that node's out-neighbors to the list. The algorithm proceeds until some stopping condition is met,  at which point it returns the closest node to $q$ seen so far. Various stopping conditions lead to algorithms that are popular in the literature including the ubiquitous \emph{beam search} method \citep{SubramanyaDevvritKadekodi:2019}. See \Cref{sec:setup} for detailed pseudocode. The query-time of a graph-based method is the time required to execute generalized greedy search and return an approximate nearest neighbor for the query, i.e. to find some $\tilde{x}$ such that, for approximation factor $c \geq 1$, $\|\tilde{x} - q\|_2 \leq c \cdot \min_{x\in X} \|x - q\|_2$.\footnote{In this work, we do not consider preprocessing time used to build the search graph $G$: our lower bounds will hold for any graph-based method, even one that is allowed arbitrarily large preprocessing time.}

\subsection{Lower Bound for Graph-Based Search}\label{sec:lbIntro} 
Our first result explains the lack of progress on worst-case guarantees for graph-based search.
We prove that no algorithm implemented in the model above (with any $G$, any $s$, and any stopping condition) can solve the approximate search problem under the Euclidean distance in worst-case sublinear time, for any approximation factor. This is in stark contrast to what is achievable with space-partitioning methods like LSH.

\begin{theorem}\label{thm:lbIntro} No graph-based method can solve $c$-approximate nearest neighbor search under the Euclidean distance for any finite approximation factor $c$ in $o(nd)$ time.
\end{theorem}
\Cref{thm:lbIntro} is proven in \Cref{sec:lower} by constructing a hard instance of points $X = \{x_1,\ldots,x_n\} \subset \R^n$, and showing that, no matter the choice of search graph,  greedy search is inefficient even for queries $q \in X$. This yields a lower bound for any approximation factor $c$ since, for such ``in distribution'' queries, $\min_{x \in X} \norm{x-q}_2 = 0$ and thus approximate and exact search are equivalent.

Our hard instance has $x_i = a_i v_i$, where $v_1,\ldots,v_n \in \R^n$ is an arbitrary set of orthonormal vectors, and $a_1,\ldots,a_n \in \R$ are scalars satisfying $0 \le a_1 < a_2 <\ldots < a_n$. 
Since all points in $X$ are mutually orthogonal, $\norm{x_i - x_j}_2^2 = \norm{x_i}_2^2 + \norm{x_j}_2^2 = a_i^2 + a_j^2$. It follows that $x_1$ is the closest neighbor to all other points, $x_2$ is the second closest neighbor, and so on. For this reason, if we query a point $q \in X$, greedy search will traverse $X$ in a fixed order, independent of the query, until it finds $q$. The search therefore requires $\Omega(n)$ distance computations on average for  $q \in X$.
This proof can be directly extended to a lower bound for search under cosine similarity, and to hierarchical graph-based search methods like HNSW. 

\smallskip

\noindent\textbf{Low-Dimensions.} 
\Cref{thm:lbIntro} rules out graph-based search methods that achieve worst-case sublinear query time akin to the $\tilde O(dn^{1/c^2 +o(1)})$ runtime of LSH \citep{AndoniIndyk:2008}. However, our hard instance requires $d = n$. Thus, the theorem does not rule out runtimes of the form $O(d^\beta n^{\alpha})$ for $\beta > 1$ and $\alpha < 1$ -- i.e., it leaves open the possibility of sublinear time search when the dimension $d$ is sufficiently small. The setting of small $d$ is particularly important given that, using  dimensionality reduction methods like Johnson-Lindenstrauss random projection \citep{johnson1984extensions,ailon2006approximate}, we can always project $X$ into $d = O(\epsilon^{-2} \log n)$ dimensional space, while distorting distances  by at most a $(1+\epsilon)$ factor.

%%% Chris version
We partially address the low-dimensional setting by proving a query time lower bound of $\Omega(n^{1/2-\gamma})$ for any constant $\gamma > 0$, even when $d = O(\log n)$. This lower bound is in a  more restricted setting, applying to methods that use \emph{pure greedy search}, meaning that search terminates as soon as the node removed from the search list fails to make progress towards $q$. Moreover, the lower bound holds for methods that succeed for all queries  $q \in X$ (instead of for any query with high probability). Nevertheless, it gives evidence that graph-based search cannot match LSH, even in low dimensions. Recall that for $q \in X$, the optimal nearest neighbor has distance $0$, so LSH can achieve arbitrarily small query time by taking the approximation factor $c$ sufficiently large. Our $\Omega(n^{1/2-\gamma})$ lower bounds shows this is not possible with pure greedy graph search.

%In particular, unlike LSH, we cannot drive the query complexity to be an arbitrarily small polynomial in $n$ by taking $c$ sufficiently large.

\begin{theorem}\label{thm:lbLow} There is a set of $n$ points $X \subset \R^d$ for $d = O(\log n)$ on which no graph-based method using pure greedy search can succeed for all queries $q\in X$ in $o(dn^{1/2-\gamma})$ time for any fixed $\gamma > 0$.
\end{theorem}

The hard instance for \Cref{thm:lbLow} is simply a random point set in $O(\log n)$ dimensions. We build on a recent average degree lower bound for the related problem of constructing \emph{navigable graphs} \citep{DiwanGouMusco:2024}. In short, we show that the starting node in $G$ must have degree $\Omega(n^{1/2-\gamma})$ for greedy search to succeed on all queries in $X$. Therefore, even running the very first search step already takes $\Omega(dn^{1/2-\gamma})$ time.

We remark that some graph-based methods achieve worst-case query time of roughly $2^{O(d)}$  \citep{AryaMount:1993,Clarkson:1994,IndykXu:2023,lu2025proximity}. \Cref{thm:lbLow} shows that, at least for pure greedy search, the exponential dependence on $d$ is inherent: as soon as $d = \Theta(\log n)$ the runtime must depend polynomially on $n$. 
Extending such a result to generalized greedy search is an interesting open question. Alternatively, it would be interesting to obtain runtime bounds, e.g., of the form $O(d^\beta \sqrt{n})$, which does not match LSH, but can still be significantly sublinear, even for $d = \Omega(\log n)$.

\vspace{-.5em}
\enlargethispage{\baselineskip} % lets the footnote to Thm 3 fit on page 3
\subsection{Sublinear Runtime Guarantees via Steiner Points} In spite of the lower bounds in \Cref{thm:lbIntro,thm:lbLow}, we argue that a simple modification of existing graph-based search frameworks does allow them to match the theoretical guarantees of LSH. In particular, graph-based search becomes significantly more powerful when the search graph $G$ is allowed to include \emph{Steiner nodes} that do not correspond to points in the dataset. We show:

\begin{theorem}\label{thm:ubIntro} There exists a polynomial time algorithm that, given any $n$-point dataset $X \in \R^d$, constructs a search graph $G$ over $X$ and $\tilde O(n^{1+1/(c-1)})$ additional Steiner points, such that generalized greedy search on $G$ run for $\tilde O(n^{1/(c-1)})$ steps starting from  an appropriate start node, solves $c$-approximate Euclidean nearest-neighbor search in $\tilde O(d n^{1/(c-1)})$ time, with high probability.\footnote{Here, $\tilde O(\cdot)$ hides logarithmic factors in the input parameters, the failure probability, and the \emph{aspect ratio} $\Delta = \frac{\max_{i \neq j} \norm{x_i-x_j}_2}{\min_{i \neq j} \norm{x_i-x_j}_2}$. Mild dependence on this aspect ratio is standard for LSH-based approximate nearest neighbor search results. The big-O notation also hides constants that depend on the approximation factor $c$.}
\end{theorem}

We prove \Cref{thm:ubIntro} by showing that a Euclidean norm LSH method of \cite{IM98} and \cite{DIIM04} can be \emph{simulated} using greedy search on a graph $G$ that includes a carefully chosen set of Steiner points. While the full construction is somewhat intricate, the core idea is simple. To explain it, we first recall the general approach taken in LSH methods. The central object in such methods is a locality sensitive hash function, which has the property that two points are more likely to have the same hash value if they are close to each other. This allows one to restrict the nearest neighbor search problem: we hash the query point $q$ and just search for a near neighbor among the $\ll n$ points in the dataset that hash to the same value. Repeating this restricted search a sufficient number of times, and returning the best result from the repetitions, ensures that an approximate nearest neighbor to $q$ is returned with high probability. 

\cite{DIIM04} employs an LSH function based on random projecting points to a one-dimensional line, $v \in \R^d$. This line is divided into equal-length segments (``buckets''), and the hash value of a point is the index of the bucket that it lands in. Without loss of generality, $v$ can be shifted to lie far away and on one side of all points in $X$ without changing the hash function. Referring to the example in \Cref{fig:randomProjection}, $x_2$ and $x_{12}$ hash to bucket 1, $x_1, x_3,$ and $x_4$ to bucket 2, etc. 
% Nearby points are more likely to land in the same bucket. 

Our key observation is that this hash function can be implemented using greedy search with Steiner points. Here we give a stylized explanation of how to do so -- our actually construction is somewhat more intricate. We place Steiner points $s_1, s_2,\ldots,$ along $v$ at the center of each bucket. 
By the Pythagorean theorem, the closest Steiner point to any $x_j \in X$ is the point at the center of the bucket that $x_i$ hashes to. We add an edge to $x_i$ from that Steiner point.
We then place an additional starting node, $s_0$, far away from $X$. We add out-edges from $s_0$ to all other Steiner points. In the first step of greedy search, we will navigate from $s_0$ to the Steiner point, $s_i$, corresponding to the query $q$'s hash value. Search will then proceed to points in $X$ that also hashed to the same bucket (and are thus connected to $s_i$). 
By locating $v$ far away from the point set, we ensure that for all non-trivial queries $q$ (i.e., for which not all points are $c$-approximate nearest neighbors), any $s_i$ is further from $q$ than any $x_j$. This ensures that greedy search properly progresses from $s_i$ to the database points, instead of exploring other Steiner nodes.

The above construction, which is visualized in \Cref{fig:randomProjection}, is the key ingredient used in implementing an LSH search method within the framework of graph-based search. The full implementation requires further ideas. For example, rather than directly connecting the starting node to a Steiner point for each hash bucket, we need to leverage additional Steiner points to implement a more efficient binary search for the hash bucket. We also chain together several basic search graphs to implement repetition, which is required to find an approximate nearest neighbor with high probability. 
Details are given in \Cref{sec:steiner}.

%\paragraph{Proof intuition.}
%The lower bounds come from the fact that, without extra nodes, the graph only has the original data points to work with, and those points may simply not contain the right structure for fast search. Steiner points let us add exactly that missing structure. The idea is to build auxiliary nodes that play the role of an indexing scheme, so the search can quickly route the query into a much smaller candidate set. Once that structure is in place, we can prove a sublinear runtime bound instead of just hoping the graph geometry will guide the search.

\begin{figure}[h]
\centering
\raisebox{-0.5\height}{\begin{tikzpicture}[scale=.77, >=Latex]

% ---- parameters ----
\def\theta{28}
\def\L{4.1}
\def\tick{0.10}

% ---- axes ----
% \draw[<->] (-3.2,0) -- (3.2,0) node[right] {$x$};
% \draw[<->] (0,-2.6) -- (0,2.6) node[above] {$y$};

% ---- direction unit vector u=(cos,sin) and perpendicular up ----
\pgfmathsetmacro{\ux}{cos(\theta)}
\pgfmathsetmacro{\uy}{sin(\theta)}
\pgfmathsetmacro{\upx}{-sin(\theta)}
\pgfmathsetmacro{\upy}{ cos(\theta)}

% ---- the line representing v (both directions, arrows at both ends) ----
\draw[<->] ({-\L*\ux},{-\L*\uy}) -- ({\L*\ux},{\L*\uy}) node[above right, xshift=-1pt,yshift=-2pt,] {$v$};;

% --------------------------------------------------------------------
% Unrolled tick marks (edit each block independently)
% --------------------------------------------------------------------

% open paren
\def\tParOne{-3.7}
\node[rotate=\theta] at ({\tParOne*\ux},{\tParOne*\uy}) {$($};
% Tick 1
\def\tOne{-3}
\node[red] at ({\tOne*\ux},{\tOne*\uy}) {\tiny $\bigstar$};
\node[above left, red] at ({\tOne*\ux + 0.00*\upx},{\tOne*\uy + 0.00*\upy}) {$s_1$};
% closeParen
\def\tParTwo{-2.3}
\node[rotate=\theta] at ({\tParTwo*\ux},{\tParTwo*\uy}) {$)$};

% open paren
\def\tParOne{-2.2}
\node[rotate=\theta] at ({\tParOne*\ux},{\tParOne*\uy}) {$($};
% Tick 2
\def\tTwo{-1.5}
\node[blue] at ({\tTwo*\ux},{\tTwo*\uy}) {\tiny $\bigstar$};
\node[above left, blue] at ({\tTwo*\ux + 0.00*\upx},{\tTwo*\uy + 0.00*\upy}) {$s_2$};
% closeParen
\def\tParTwo{-.8}
\node[rotate=\theta] at ({\tParTwo*\ux},{\tParTwo*\uy}) {$)$};

% open paren
\def\tParOne{-.7}
\node[rotate=\theta] at ({\tParOne*\ux},{\tParOne*\uy}) {$($};
% Tick 3
\def\tThree{0}
\node[orange] at ({\tThree*\ux},{\tThree*\uy}) {\tiny $\bigstar$};
\node[above left, orange] at ({\tThree*\ux + 0.00*\upx},{\tThree*\uy + 0.00*\upy}) {$s_3$};
% closeParen
\def\tParTwo{.7}
\node[rotate=\theta] at ({\tParTwo*\ux},{\tParTwo*\uy}) {$)$};

% open paren
\def\tParOne{.8}
\node[rotate=\theta] at ({\tParOne*\ux},{\tParOne*\uy}) {$($};
% Tick 4
\def\tFour{1.5}
\node[teal] at ({\tFour*\ux},{\tFour*\uy}) {\tiny $\bigstar$};
\node[above left, teal] at ({\tFour*\ux + 0.00*\upx},{\tFour*\uy + 0.00*\upy}) {$s_4$};
% closeParen
\def\tParTwo{2.2}
\node[rotate=\theta] at ({\tParTwo*\ux},{\tParTwo*\uy}) {$)$};

% open paren
\def\tParOne{2.3}
\node[rotate=\theta] at ({\tParOne*\ux},{\tParOne*\uy}) {$($};
% Tick 5
\def\tFive{3}
\node[magenta] at ({\tFive*\ux},{\tFive*\uy}) {\tiny $\bigstar$};
\node[above left, magenta] at ({\tFive*\ux + 0.00*\upx},{\tFive*\uy + 0.00*\upy}) {$s_5$};
% closeParen
\def\tParTwo{3.7}
\node[rotate=\theta] at ({\tParTwo*\ux},{\tParTwo*\uy}) {$)$};

% --------------------------------------------------------------------
% Points and per-projection colors (edit these lists)
% --------------------------------------------------------------------
\def\npts{12}

% x_i positions as x/y pairs
\def\points{
  -2.1/ 1.2,
  -3.1/-1,
  -1.4/ 0.6,
  -1.5/-1,
  -0.4/ 1.3,
   0.1/-0.9,
   1.2/ 1.5,
   .5/0,
   2.5/ .5,
   1/.4,
   2.8/ 1.5,
  -3/ 0.4
}

% Color for each projection point p_i (must have 12 entries)
% Use any xcolor names (e.g. red, blue!60, black, gray, teal, etc.)
\def\pcolors{
  blue,
  red,
  blue,
  blue,
  orange,
  orange,
  teal,
  orange,
  magenta,
  teal,
  magenta,
  red}

% realized I need to shift the points
\def\dx{-3*\uy}
\def\dy{3*\ux}

% Build a shifted copy of \points into \pointsShifted
\def\pointsShifted{}%
\foreach \px/\py in \points {%
  \pgfmathsetmacro{\sx}{\px+\dx}%
  \pgfmathsetmacro{\sy}{\py+\dy}%
  \ifx\pointsShifted\empty
    \xdef\pointsShifted{\sx/\sy}%
  \else
    \xdef\pointsShifted{\pointsShifted,\sx/\sy}%
  \fi
}%

% Overwrite \points so the rest of your code doesn't change
\let\points\pointsShifted

% Tick 0
\node[black] at ({-.6*\dx},{-.6*\dy}) {\tiny $\bigstar$};
\node[above left, black] at ({-.6*\dx},{-.6*\dy}) {$s_0$};

% ---- draw points and orthogonal projections, naming coordinates ----
\foreach [count=\i] \px/\py in \points {
  % name the original point
  \coordinate (x\i) at (\px,\py);

  % scalar projection t = <p,u>
  \pgfmathsetmacro{\tproj}{(\px*\ux + \py*\uy}

  % projected point = t*u, name it p<i>
  \coordinate (p\i) at ({\tproj*\ux},{\tproj*\uy});

  % dotted gray projection segment
  \draw[densely dotted, gray] (x\i) -- (p\i);

  % original point
}

% ---- color and draw each projection point (unrolled via a second loop) ----
% This keeps the projection-point styling easy to control separately.
\foreach [count=\i] \c in \pcolors {
  \fill[\c] (p\i) circle (1.5pt);
  \fill[\c] (x\i) circle (1.5pt);
  \node[anchor=south west, xshift=-2pt, yshift=-2pt] at (x\i) {$x_{\i}$};
}

\end{tikzpicture}}%
\hspace{4em}
\raisebox{-0.5\height}{\begin{tikzpicture}[>=Latex, scale=1]

%------------------------------------------------
% Coordinates (compact layout, about 2 inches high)
%------------------------------------------------

% Left layer
\coordinate (s0) at (0,0);

% Middle layer
\coordinate (s1) at (1.5,  1.6);
\coordinate (s2) at (1.5,  0.8);
\coordinate (s3) at (1.5,  0.0);
\coordinate (s4) at (1.5, -0.8);
\coordinate (s5) at (1.5, -1.6);

% Right layer: x_2, x_12, x_4, x_1, x_3, x_6, x_5, x_8, x_10, x_7, x_9, x_11
\coordinate (x2)  at (3.3,  2.2);
\coordinate (x12) at (3.3,  1.8);
\coordinate (x4)  at (3.3,  1.4);
\coordinate (x1)  at (3.3,  1.0);
\coordinate (x3)  at (3.3,  0.6);
\coordinate (x6)  at (3.3,  0.2);
\coordinate (x5)  at (3.3, -0.2);
\coordinate (x8)  at (3.3, -0.6);
\coordinate (x10) at (3.3, -1.0);
\coordinate (x7)  at (3.3, -1.4);
\coordinate (x9)  at (3.3, -1.8);
\coordinate (x11) at (3.3, -2.2);

%------------------------------------------------
% Draw directed edges
%------------------------------------------------

% s0 -> s1,...,s5
\draw[->] (s0) -- (s1);
\draw[->] (s0) -- (s2);
\draw[->] (s0) -- (s3);
\draw[->] (s0) -- (s4);
\draw[->] (s0) -- (s5);

% s1 -> x2, x12
\draw[->] (s1) -- (x2);
\draw[->] (s1) -- (x12);

% s2 -> x1, x3, x4
\draw[->] (s2) -- (x1);
\draw[->] (s2) -- (x3);
\draw[->] (s2) -- (x4);

% s3 -> x6, x5, x8
\draw[->] (s3) -- (x6);
\draw[->] (s3) -- (x5);
\draw[->] (s3) -- (x8);

% s4 -> x10, x7
\draw[->] (s4) -- (x10);
\draw[->] (s4) -- (x7);

% s5 -> x9, x11
\draw[->] (s5) -- (x9);
\draw[->] (s5) -- (x11);

%------------------------------------------------
% Draw s_i stars and labels
%------------------------------------------------

\node[black] at (s0) {{\tiny $\bigstar$}};
\node[anchor=east, xshift=0pt] at (s0) {$s_0$};

\node[red] at (s1) {{\tiny $\bigstar$}};
\node[anchor=east, xshift=-2pt, text=red] at (s1) {$s_1$};

\node[blue] at (s2) {{\tiny $\bigstar$}};
\node[anchor=east, xshift=-2pt, text=blue] at (s2) {$s_2$};

\node[orange] at (s3) {{\tiny $\bigstar$}};
\node[anchor=east, xshift=-2pt, yshift=-4pt, text=orange] at (s3) {$s_3$};

\node[teal] at (s4) {{\tiny $\bigstar$}};
\node[anchor=east, xshift=-2pt, text=teal] at (s4) {$s_4$};

\node[magenta] at (s5) {{\tiny $\bigstar$}};
\node[anchor=east, xshift=-2pt, text=magenta] at (s5) {$s_5$};

%------------------------------------------------
% Draw x_i black dots and labels
%------------------------------------------------

\fill[red] (x2)  circle (1.2pt);   \node[anchor=west, xshift=2pt] at (x2)  {$x_2$};
\fill[red] (x12) circle (1.2pt);   \node[anchor=west, xshift=2pt] at (x12) {$x_{12}$};
\fill[blue] (x4)  circle (1.2pt);   \node[anchor=west, xshift=2pt] at (x4)  {$x_4$};
\fill[blue] (x1)  circle (1.2pt);   \node[anchor=west, xshift=2pt] at (x1)  {$x_1$};
\fill[blue] (x3)  circle (1.2pt);   \node[anchor=west, xshift=2pt] at (x3)  {$x_3$};
\fill[orange] (x6)  circle (1.2pt);   \node[anchor=west, xshift=2pt] at (x6)  {$x_6$};
\fill[orange] (x5)  circle (1.2pt);   \node[anchor=west, xshift=2pt] at (x5)  {$x_5$};
\fill[orange] (x8)  circle (1.2pt);   \node[anchor=west, xshift=2pt] at (x8)  {$x_8$};
\fill[teal] (x10) circle (1.2pt);   \node[anchor=west, xshift=2pt] at (x10) {$x_{10}$};
\fill[teal] (x7)  circle (1.2pt);   \node[anchor=west, xshift=2pt] at (x7)  {$x_7$};
\fill[magenta] (x9)  circle (1.2pt);   \node[anchor=west, xshift=2pt] at (x9)  {$x_9$};
\fill[magenta] (x11) circle (1.2pt);   \node[anchor=west, xshift=2pt] at (x11) {$x_{11}$};

\end{tikzpicture}}
\caption{We illustrate the basic building block of our graph-based simulation of the projection-based LSH method of \cite{DIIM04}. Steiner points $s_1, \ldots, s_5$ are placed along a random line, $v$, and one additional Steiner point, $s_0$, is placed far away from the dataset. Running greedy search on the directed search graph on the right, starting at $s_0$, implements the core LSH primitive: for a query $q$, it will first search for nearest neighbors among all points that project into the same bucket on $v$.}\label{fig:randomProjection}
\end{figure}
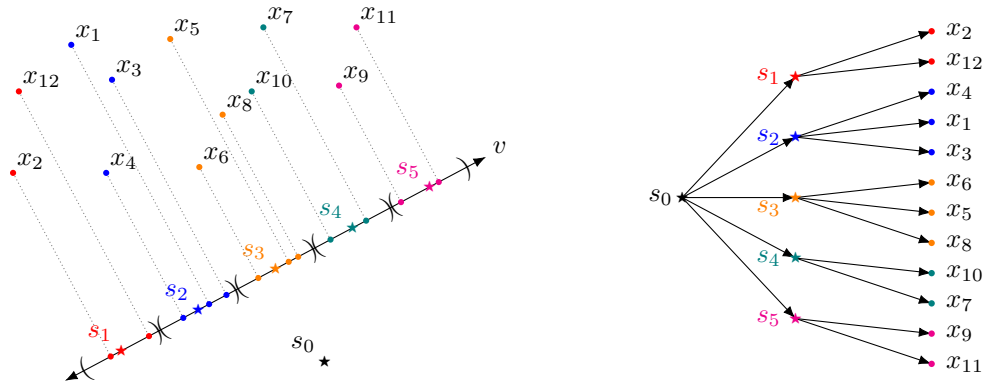

\subsection{Empirical Evaluation of Steiner Point-Based Methods}

In combination, the lower bounds of \Cref{thm:lbIntro,thm:lbLow} and the upper bound of 
\Cref{thm:ubIntro} show that, at least in theory, Steiner points can significantly increase the power of graph-based nearest neighbor search.
%when augmented with Steiner nodes, greedy graph-based search can simulate LSH, and thus achieve sublinear query time for approximate nearest neighbor search in high-dimensional Euclidean space. This is in sharp contrast to the situation without Steiner nodes, where the lower bounds of \Cref{thm:lbIntro,thm:lbLow} rule this out.
In \Cref{sec:experiments}, we explore the practical implications of this finding. We augment ParlayANN's \citep{ManoharShenBlelloch:2024} implementation of DiskANN with Steiner points chosen in a very simple way: we cluster the dataset into $k = cn$ clusters for some small constant $c < 1$ using $k$-means, and add the cluster centroids as Steiner points. We build a search graph using the standard DiskANN algorithm, but over the dataset augmented with these Steiner points. We find that even this very simple method demonstrates improved performance on several standard data sets. 
We note that other works have considered the direct or indirect use of Steiner points during search graph construction, often to improve performance on out-of-distribution queries \citep{jaiswal2022ood,chen2024roargraph,yue2026efficient,BaranchukPersiyanovSinitsin:2019}. 
Our findings suggest that further explorations in this direction may be worth consideration. %leave open the question of how Steiner points 

\section{Problem Formalization and Greedy Search}
\label{sec:setup}
Before moving on to our main technical sections, we restate the problem studied in this work more formally, and provide pseudocode for the generalized greedy search algorithm that our lower bounds cover. 

\smallskip

\noindent\textbf{Nearest Neighbor Search.} In the approximate nearest neighbor (ANN) problem, we are given a dataset $X = \{x_1, \ldots, x_n\} \subset \R^d$ and a distance function $D:\R^d \times \R^d$. For a query $q \in \R^d$, the goal is to approximately minimize $\min_{i\in 1, \ldots n}D(x_i,q)$.  Typical choices of $D$ include the Euclidean distance $D(x,y) = \|x-y\|_2 = (\sum_{j=1}^d (x[j] - y[j])^2)^{1/2}$ or inverse cosine similarity $D(x,y) = 1 - \frac{\ip{x}{y}}{\norm{x}_2\norm{y}_2}$. Practical systems typically aim to return the top $k$ closest vectors to $q$ from $X$ for $k \geq 1$. For simplicity, this paper focuses on the $k = 1$ case, in line with most prior theoretical work in the area.

\smallskip

\noindent\textbf{Greedy Graph-Based Search.} We consider methods that solve the nearest neighbor search problem by running generalized greedy search (\Cref{alg:gen_greedy_search}) on a directed search graph $G = (\{1, \ldots, n\}, E)$. $G$ has one node corresponding to every data point in $X$. We use the notation $\mathcal{N}_G(i)$ to denote the set of out-neighbors of a node $i$ in $G$. $\mathcal{N}_G(i)$ contains all $j$ for which the ordered pair $(i,j) \in E$.

\begin{algorithm}[h!]
    \caption{Generalized Greedy Search}
    \label{alg:gen_greedy_search}
    \begin{algorithmic}[1]
        \Require Dataset $X = \{x_1, \ldots, x_n\}$, search graph $G$, starting node $s$, distance function $D$, query $q$.
        % \Ensure Approximate nearest neighbor $\tilde{x} \in X$.
        \algrule
        \State Initialize min-priority queue $\mathcal{C}$.
        \Comment{Elements are node IDs, priorities are distances to $q$.}
        \State Initialize dictionary $\mathcal{D}$. \Comment{Used to keep track of already explored nodes.}
        \State Initialize $\tilde{x} = x_s$. Insert $(s, D(q, x_s))$ into $\mathcal{C}$ and $\mathcal{D}$.  
        \While{\textbf{\textcolor{red}{not [termination condition]}}} \label{line:stop_condition}
            \State $(j,D(q,x_j)) \gets \textrm{extractMin}(\mathcal C)$. 
            \Comment{Pop the node closest to $q$.}
            \If{$D(x_j,q) \leq D(\tilde{x},q)$} $\tilde{x}\gets x_j$. \Comment{Update guess for nearest neighbor.}\label{set_x_tilde}
                
            \EndIf
            \State For all $k \in \mathcal{N}_G(j)$, if $k$ is not in $\mathcal{D}$, insert $(k,D(q,x_k))$ into $\mathcal{C}$ and $\mathcal{D}$.
        \EndWhile
        \State \Return $\tilde{x}$.
    \end{algorithmic}
\end{algorithm}

Generalized greedy captures a variety of algorithms depending on how the stopping condition is set on Line \ref{line:stop_condition}. One option is to terminate if $D(x_j, q) \geq D(\tilde{x}, q)$, which means we failed to make one-step progress towards $q$. This approach is often simply called greedy search. More common in practice is to use beam search with beam width $b$, which is equivalent to terminating only if $D(x_j, q) \geq D(x_{j_b}, q)$, where $x_{j_b}$ is the $b^{\text{th}}$ closest node to $q$ among all nodes in $\mathcal{D}$. We refer the reader to \cite{Al-JazzaziDiwanGou:2025} for a discussion of various termination conditions, emphasizing that our lower bounds hold for \emph{any such condition}. Indeed, for the purpose of search with worst-case query complexity, one natural approach is simply to terminate after a fixed number of computational steps, $T$. %, have been executed.

With \Cref{alg:gen_greedy_search} in place, we can state our formal definition of graph-based ANN search:  

\begin{definition}[Graph-Based $c$-Approximate Nearest Neighbor Search]
    \label{def:ann_graph}
Let $c \geq 1$ be an approximation factor. Consider a potentially randomized algorithm, $\mathcal{A}$, that, given a dataset $X\subset\R^d$, returns a search graph $G$, starting node $s$, and termination condition COND. We say that $\mathcal{A}$ solves the $c$-approximate nearest neighbor search problem with query time $T$ and failure probability $\delta\in (0,1)$ if, for any fixed $q \in \R^d$, with probability at least $1-\delta$, \Cref{alg:gen_greedy_search} run using $G$, $s$ and  COND returns a solution $\tilde{x}$ in time $\leq T$ satisfying
$D(\tilde{x},q) \leq c \cdot \min_{x\in X} D(x,q)$.
\end{definition}
Generalized greedy search and \Cref{def:ann_graph} naturally extend to search graphs whose node set is $X \cup S$ for some set of Steiner nodes $S$. The nodes in $S$ are treated identically to those in $X$. The only change is that the output $\tilde x$ is the closest \emph{non-Steiner node} to $q$ that has been popped off the queue. If no node in $X$ is popped before termination, a default output is returned, e.g., $\tilde x = x_1$.

\section{Lower Bounds for Graph-Based Search}
\label{sec:lower}
Our main research question is whether algorithms satisfying  \Cref{def:ann_graph} with runtime $T= o(nd)$ exist. 
Our first result answers this question in the negative: for worst-case inputs, no graph-based search method with finite approximation factor can solve approximate nearest neighbor search in sublinear time when $D$ is the Euclidean distance. 
% As we will discuss, our proof easily extends to related settings, e.g., when $D$ is the inverse cosine similarity, and when the search method can use a hierarchical, rather than flat, search graph. 
As discussed, our hard instance is a set of $n$ vectors in $\R^n$. We prove a weaker lower bound (\Cref{thm:lbLow}) for $O(\log n)$-dimensional points in \Cref{sec:low-dim-single-source-lb}.

\begin{theorem}[\Cref{thm:lbIntro} restated]\label{thm:main_lower}
There is a set of $n$ points, $X \subset \R^{n}$, such that no graph-based method satisfies \Cref{def:ann_graph} with  $D$ being the Euclidean distance, $c$ finite, $\delta \le 1/2$, and $T = o(nd)$.
%and a set of queries $\mathcal Q \subset \R^n$ such that \Cref{alg:gen_greedy_search} run with any search graph $G$ over $\mathcal{X}$, any starting node $s$, and any termination condition, requires expected runtime $\Omega(n)$ to solve \Cref{prob:ann} under the Euclidean distance for any finite approximation factor $C$ on a random query from $\mathcal Q$.
\end{theorem}
\begin{proof}
As sketched in \Cref{sec:lbIntro}, our hard instance $X = \{x_1,\ldots,x_n\}$ has $x_i = a_i v_i$, where  $v_1,\ldots,v_n \in \R^n$ is an orthonormal basis and $a_1,\ldots,a_n \in \R$ satisfy $0 \le a_1  <\ldots < a_n$.
We will show that, for a query $q$ drawn uniformly at random from $X$, any graph-based method fails to output $q$ with probability at least $1/2$ in $o(nd)$ time. This yields the lower bound since, when $q \in X$, solving the approximate search problem for any finite approximation factor $c$ requires returning $q$ itself.

%Our hard instance will have $\mathcal Q = \mathcal X$ -- i.e., our queries are in our dataset. Thus, to solve \Cref{prob:ann} for any finite approximation factor, given a query $q \in \mathcal{X}$, the algorithm must output $q$ itself. We will argue that this requires $\Omega(n)$ time on average over a random query $q \in \mathcal X$.

Since all points in $\mathcal X$ are mutually orthogonal, $\norm{x_i - x_j}_2^2 = \norm{x_i}_2^2 + \norm{x_j}_2^2 = a_i^2 + a_j^2$ for all $i,j \in [n]$. That is, for any $x_i, x_j \neq x_k$, $\norm{x_i - x_k}_2 < \norm{x_j - x_k}_2$ whenever $a_i^2 < a_j^2$, i.e., whenever $i <j$. So, for any query $q \in  X$, at each step of greedy search, if we do not directly pop $q$ off the search queue $\mathcal C$ (and succeed in outputting $q$), we simply pop off the point with the smallest index. I.e., the execution of the algorithm is identical for any query $q \in {X}$, up to the point that $q$ is added to $\mathcal C$. %Since the algorithm is simply doing a non-adaptive linear scan over the data, it's expected runtime on a random $q \in X$ is $\Omega(nd)$.

Formally, let $r_i$ be the iteration of the while loop in \Cref{alg:gen_greedy_search} that $x_i$ is added to $\mathcal C$ when $q = x_i$. If $x_i = s$ is the starting node, let $r_i = 0$. Since nodes are added to $\mathcal{C}$ in the same order regardless of the query, the number of distance computations required to find query $q = x_i$ is equal to $c_i = |\{j \in [n]: r_j \le r_i\}|$. A simple but key observation is that $c_i \ge n/4$ for at least $3n/4$ points in $X$.

Now, if the algorithm has success probability $1/2$, with probability $1/2$ over the randomness used in selecting $G$ and $s$, greedy search must succeed when querying at least half of the points from $X$. At least $3n/4 - n/2 = n/4$ of those points have $c_i \ge n/4$, and thus require $\Omega(nd)$ time to find with greedy search.  
 It follows that the expected (and thus, worst-case) runtime of greedy search is $\Omega(nd)$ for any graph-construction methods with success probability $\geq 1/2$.
\end{proof}

\smallskip

\noindent \textbf{Extensions.} 
We remark that \Cref{thm:main_lower} in fact extends to any algorithm that only accesses the distance function through comparisons of the form $D(q,x) \le D(q,y)$. In our hard instance, the ordering of distances from the query $q$ to points in $X$ is independent of $q$, so any such algorithm  executes identically for all $q$, until $q$ is returned. Our same argument thus yields an $\Omega(nd)$ lower bound for methods based on \emph{hierarchical search graphs}, like HNSW. These methods perform greedy search, but on a layered graph: each layer is a search graph over a subset of points $L_i \subseteq X$ with $L_1 \subseteq L_2 \subseteq L_3,\ldots$.
They first use generalized greedy search to find a set of near neighbors to $q$ in $L_1$. They then use these points as starting points for a greedy search in $L_2$, etc. As in flat graph methods, the distance function is only used for comparisons of the form $D(q,x) \le D(q,y)$.

The lower bound also easily extends to the case when $D$ is the inverse cosine similarity: $D(x,y) = 1 - \frac{\ip{x}{y}}{\norm{x}_2\norm{y}_2}$. In this case, we append an additional coordinate to each point (so the points are now in $\R^{n+1})$, with value $\sqrt{1-a_i^2}$. This ensures that all points are unit norm and $D(x_i,x_j) = 1-\sqrt{(1-a_i^2)(1-a_j^2)}$. Again,  $D(x_i,x_k) < D(x_j,x_k)$ exactly when $i < j$, and so the search order is independent of the query and the lower bound proof goes through unchanged. Similar constructions give lower bounds, e.g., for different $p$-norms. % and other popular distance metrics.

\section{Sublinear Search with Steiner Points}\label{sec:steiner}

We next prove \Cref{thm:ubIntro}, which shows that, with the ability to add Steiner points to the search graph, we can implement an LSH-based ANN data structure with generalized greedy search, skirting the lower bounds of \Cref{sec:lower}. In particular, we implement the original LSH-based algorithm of \cite{IM98} using the Euclidean LSH function of \cite{DIIM04}. While we focus on this particular implementation for simplicity, we believe that similar ideas likely allow for the simulation of related approaches, such as those achieving refined bounds for Euclidean distance \citep{AndoniIndyk:2008} or bounds for other similarity metrics, like cosine similarity \citep{charikar2002similarity}.

\cite{IM98} solves the approximate nearest neighbor search problem by solving a logarithmic number of simpler ``fixed radius'' search problems. In particular, for approximation factor $c \ge 1$ and radius $R$, the $(R,c)$-near-neighbor problem asks us to
return some $\tilde x \in X$ with $\norm{\tilde x - q}_2 \le cR$ if $\min_{x \in X} \norm{x-q}_2 \le R$. If $\min_{x \in X} \norm{x-q}_2 > cR$,  we must return FAIL. If $\min_{x \in X} \norm{x-q}_2 \in (R,cR]$, we can exhibit either behavior. 
Letting $\Delta = \frac{\max_{i \neq j} \norm{x_i-x_j}_2}{\min_{i \neq j} \norm{x_i-x_j}_2}$ be the \emph{aspect ratio} of the dataset $X$,  $c$-approximate ANN search can be reduced to  $O(\log \Delta)$ fixed-radius instances with geometrically spaced radii \citep{IM98,Har01}. We sketch how this reduction is implemented via generalized greedy search on $G$ at the end of this section, with a full proof in \Cref{app:steiner-lsh-to-ann}. Here, we start by focusing on  fixed radius search.

\subsection{Simulating Fixed Radius LSH with Generalized Greedy Search} 

To solve the $(R,c)$-near-neighbor problem, we use an LSH function introduced by \cite{DIIM04}.

\[
    h_{a,b}(x)=
    \left\lfloor
        \frac{\langle a,x\rangle+b}{w}
    \right\rfloor ,
\]
where $a\in\mathbb R^d$ has independent standard Gaussian entries and $b$ is drawn uniformly from $[0,w]$, with $w>0$ being a bucket width.
 Concatenating $k$ independent instantiations of this hash function  (i.e., with independently chosen $a$ and $b$) yields a length-$k$ hash signature for any point:
\begin{align}
    \label{eq:hash_function}
    g(x)=(h_{a_1,b_1}(x),\ldots,h_{a_k,b_k}(x)).
\end{align}
Given a query $q$, we compute the hash $g(q)$ and, if the set $X_{g(q)}=\{x\in X:g(x)=g(q)\}$ is non-empty, we examine an arbitrary point in the set. We return this point if it is within distance $cR$ of $q$.\footnote{A variant used by \citep{DIIM04} examines all points in the bucket, arguing that with high probability the total number of points examined is small. This variant gives a slight improvement from $\tilde O(d n^{1/(c-1)})$ to $\tilde O(d n^{1/c})$ query time. However, it is more complicated to simulate via graph-based search, so we focus on the simpler \cite{IM98} approach.}
The procedure is repeated $L$ times, with $L$ independent choices of $g$, outputting FAIL if no point within distance $cR$ is found in any of the $L$ repetitions.
The main result of \cite{IM98} implies that, when implemented with the LSH function of \cite{DIIM04}, this procedure obtains the following guarantee:\footnote{Formally, for an LSH family with collision probability  $\geq p_1$ for any point within radius $R$ of $q$, and $\leq p_2$ for any point outside radius $cR$, \cite{IM98} obtain a dependence on $n^\rho$ where $\rho= \frac{\log(1/p_1)}{\log(p_1/p_2)}$. As we discuss further in \Cref{app:steiner}, the $\rho$ achieved by the hash family of \cite{DIIM04} depends only on $c$.}

\begin{theorem}
\label{thm:datar-lshSimple}
For any $R>0$, $c>1$, and $\delta\in(0,1)$, there is a choice of bucket width $w$, hash repetition number $k=O(\log n)$, and table repetition number $L=O(n^{\rho}\log(1/\delta))$, where  $\rho$ is a constant depending only on $c$ with $\rho < \frac{1}{c-1}$, such that, for every fixed query $q$ with $\min_{x\in X}\norm{x-q}_2\le R$, the query procedure discussed above returns $\tilde x\in X$ satisfying $\norm{\tilde x-q}_2\le cR$ with probability at least $1-\delta$. The method has query time $\tilde O(dkL)= \tilde O(dn^{\rho}\log(1/\delta))$.
\end{theorem}

\paragraph{Simulating a Single Hash Table.} 
To simulate the LSH scheme above, we first need to show how to use graph-based greedy search to implement a single hash table. To do so, we begin by associating one Steiner node to each of the non-empty buckets in the hash table, i.e., we have a Steiner node $s_I$ for each $k$-tuple of integers $I$ such that $I = g(x_j)$ for some $x_j \in X$. There are at most $n$ such Steiner nodes. We also introduce a single ``failure'' Steiner node, $s_{F}$, which is  associated with all empty buckets in the table, i.e. all $k$-tuple of integers that do not equal $g(x_j)$ for any $j \in \{1, \ldots, n\}$. Each $s_I$ is connected with an out-edge to one arbitrary point in $X_{I}$ -- i.e., to an arbitrary point with hash value $I$.

We then build an (incomplete) binary tree of Steiner nodes, rooted at a starting node $s$, which is used to efficiently route a query to the Steiner representative, $s_{g(q)}$, or to $s_{F}$ if $q$'s bucket, $X_{g(q)}$, is empty. The internal nodes of the binary tree are chosen so that greedy search will reach the appropriate terminal or failure node in $O(k\log n)$ steps. Concretely, we show:

\begin{lemma}[Simulating one LSH table]
\label{lem:one-table-simulation-main}
Fix any \(0\le Q<\infty\)\footnote{Our construction will only work for queries with some bounded norm $Q$. However, we can set $Q$ arbitrarily large. For fixed-radius search, we can take $Q$ large enough that any query outside this radius has no neighbor within distance $R$.}, concatenated \cite{DIIM04} hash function  $g(x)=(h_{a_1,b_1}(x),\ldots,h_{a_k,b_k}(x))$,
and a dataset \(X\subset\mathbb{R}^d\) of size \(n\).
There is a directed graph \(G\) over \(X\) and $O(kn)$ additional Steiner points, with $O(kn)$ edges, and a Steiner starting node
\(s\), with the following
property.  For every query \(q\) with $\norm{q}_2 \le Q$, generalized greedy search
initiated at \(s\) and run for $T=O(k\log n)$ steps either:
\begin{enumerate}
    \item terminates at a designated Steiner failure node $s_{F}$, if \(X_{g(q)}=\varnothing\); or
    \item reaches a designated terminal node $s_{g(q)}$ if
          \(X_{g(q)}\neq\varnothing\), with a single out-edge to an arbitrary point in \(X_{g(q)}\), and then pops this out-neighbor off the search queue ($\mathcal{C}$ in \Cref{alg:gen_greedy_search}).
\end{enumerate}
All nodes in $G$ have out-degree $\le 2$. The search uses $O(T)$ distance computations and $\tilde O(Td)$ total time.
\end{lemma}

We sketch the idea behind \Cref{lem:one-table-simulation-main} below, with a full proof in  \Cref{app:steiner}. The key observation is that we can place Steiner points so that deciding between two Steiner points in a greedy search step is equivalent to evaluating a linear threshold function. Successive evaluation of these threshold functions allows us to route the search based on the value of $g(q)$. In particular, in \Cref{app:steiner-comparison-gadget} we prove:

\begin{lemma}[Steiner Comparison Gadget]
\label{lem:steiner-comparison-gadget}
Fix a nonzero vector $a\in\R^d$ and a threshold $\theta\in\R$. There exist two
Steiner points $s_L,s_R\in\R^d$ such that, for every query $q\in\R^d$,
\[
    \|q-s_L\|_2\, < \|q-s_R\|_2
    \quad \Longleftrightarrow \quad
    \langle a,q\rangle < \theta .
\]
Further, we can pick $s_L$ and $s_R$ so that $\norm{s_L}_2$ and $\norm{s_R}_2$ are arbitrarily large.
\end{lemma}
The proof of \Cref{lem:steiner-comparison-gadget} is simple: $s_L$ and $s_R$ are just different scalings of $a$, placed so that the perpendicular bisector of the line between them is exactly the hyperplane $\langle a,q \rangle = \theta$. We can move $s_L$ and $s_R$ arbitrarily far in opposite directions along $a$, keeping their midpoint fixed and preserving this bisector -- this will be useful later, in ensuring that our greedy search continues to `progress' through successive comparison gadgets.

\smallskip

\noindent\textbf{A Steiner Node Routing Tree.} With \Cref{lem:steiner-comparison-gadget} in place, we first sketch the proof of \Cref{lem:one-table-simulation-main} when $g(x) = h_{a,b}(x)$ uses just a single hash function (i.e., $k=1$). We build a binary search tree over the bucket values that occur among the
data points. Each non-leaf node has two children, $s_L$ and $s_R$, placed in order to compare  \(h_{a,b}(q) = \langle a,q\rangle+b\) to one
bucket boundary, via \Cref{lem:steiner-comparison-gadget}. Greedy search follows the child corresponding to the outcome. The leaves are the terminal Steiner nodes $s_{i}$ associated with $X_{i} = \{x \in X: h_{a,b}(x) = i\}$ when this set is nonempty, or Steiner nodes associated with  `failure intervals' between consecutive nonempty buckets. These failure leaves have a single out edge to the  failure node $s_F$. Since
there are at most \(n\) nonempty buckets, the tree has $O(n)$ leaves, and thus \(O(n)\) Steiner
routing nodes. It has depth \(O(\log n)\).

Consider running two steps of the binary search following the procedure above from some non-leaf node $s_0$. Denote the Steiner node pairs implementing the two required comparisons by $(s_{L_1},s_{R_1})$ and $(s_{L_2},s_{R_2})$, respectively. $s_0$ will have out edges to $s_{L_1}$ and $s_{R_1}$. Assuming we move left without loss of generality, $s_{L_1}$ has out edges to $s_{L_2}$ and $s_{R_2}$. 
One concern is that greedy routing may not properly progress through the two comparisons: if $\norm{s_{L_1} - q}_2\, < \norm{s_{R_1} - q}_2$, it will pop $s_{L_1}$ off the search queue in the first iteration, and add $s_{L_2}$ and $s_{R_2}$ to the queue. In the second step, we would like it to pop whichever of $s_{L_2}$ and $s_{R_2}$ is closest to $q$, thus implementing the second comparison. For this to happen, we need to ensure that it does not  `go back' and pop $s_{R_1}$. To do so, we use that the Steiner nodes in \Cref{lem:steiner-comparison-gadget} can be picked with arbitrarily large norms. By our assumption that $q$ has bounded norm (else, the $(R,c)$-near neighbor search data structure can simply return `FAIL'), we can pick $s_{L_1}$ and $s_{R_1}$  with sufficiently large norms so that they are always further from $q$ than $s_{L_2}$ and $s_{R_2}$. This forces greedy search to progress as desired -- popping either  $s_{L_2}$ or $s_{R_2}$ in the second step, if it popped $s_{L_1}$ in the first step.

Finally, to establish \Cref{lem:one-table-simulation-main}, we must consider $g(x)=(h_{a_1,b_1}(x),\ldots,h_{a_k,b_k}(x))$ with $k>1$ concatenated hash functions. We produce a routing tree for $g$ by appending one-coordinate routing trees. After the first $\ell$ coordinates, a non-failure leaf corresponds to a nonempty prefix bucket $X_I=\{x\in X:h_{a_j,b_j}(x)=i_j\text{ for }j=1,\ldots,\ell\}$, where $I=(i_1,\ldots,i_\ell)$ for $\ell<k$. This leaf becomes the root of a routing tree for $h_{a_{\ell+1},b_{\ell+1}}$ over $X_I$. The nonempty prefix buckets partition $X$, there are at most $n$ non-failure leaves among all trees at each level $\ell$. The total tree size and number of routing steps are therefore $O(kn)$ and $O(k\log n)$, respectively.

\subsection{Implementing Table Repetition}
\label{subsec:combining-tables}

\Cref{lem:one-table-simulation-main} implements a single LSH table via greedy search. To simulate the full procedure of \Cref{thm:datar-lshSimple}, we must perform lookups in $L$ independent tables. Doing so is simple. We let our global starting node $s$ be the starting node for the search graph simulating the first table via \Cref{lem:one-table-simulation-main}. Each terminal Steiner node $s_{g(q)}$, as well as the designated failure node $s_F$ for the first search graph are connected via an out edge to the starting node for a search graph simulating the second table, and so on. Dataset points have no out edges.

As we do within each search tree, we choose the lengths of  the Steiner points so that those in later trees are always closer to $q$ than those in earlier trees, assuming $\norm{q}_2\le Q$. Additionally, all data points are closer to $q$ than all Steiner points. Thus, each newly discovered data point is popped before the next table's root. After progressing through all search trees, one representative from each nonempty queried bucket has been popped. By \Cref{thm:datar-lshSimple}, if there is a data point within distance $R$ of $q$, at least one point within distance $cR$ has been popped with probability at least $1-\delta$. Overall, we have the following, proven in \Cref{app:all-table-simulation}:

\begin{theorem}[Fixed-radius Search Upper Bound]
\label{thm:fixed-radius-steiner}
Fix $R>0$, $c>1$, $\delta\in(0,1)$, and any $Q<\infty$. There is a constant $\rho$ depending only on $c$ with $\rho < \frac{1}{c-1}$ such that, for any $n$-point dataset $X\subseteq\R^d$, there is a randomized polynomial time algorithm constructing a directed graph $G$ on $X$ and $\tilde O(n^{1+\rho}\log(1/\delta))$ additional Steiner nodes, and a Steiner starting node $s$, with the following property. For every fixed $q$ with $\norm{q}_2\le Q$ and $\min_{x\in X}\norm{x-q}_2\le R$, generalized greedy search initialized at $s$ pops at least one point $x\in X$ with $\norm{x-q}_2\le cR$ within $T=\tilde O(n^{\rho} \log(1/\delta))$ steps, with probability at least $1-\delta$. All nodes in $G$ have out-degree at most two, so running generalized greedy search for $T$ steps requires $O(T)$ distance computations and takes $\tilde O(Td)$ time.
\end{theorem}

 \subsection{Reduction from ANN to Fixed Radius Search}\label{sec:fixedReduction}
 
 We conclude by describing how to use the fixed-radius search result of \Cref{thm:fixed-radius-steiner} to give our final $c$-approximate nearest neighbor search result of \Cref{thm:ubIntro}. A full proof  is given in \Cref{app:steiner-lsh-to-ann}.

We first observe that queries with sufficiently large norm, relative to the size of the dataset, are easy -- allowing us to remove the bounded norm restriction of \Cref{thm:fixed-radius-steiner}. Let $R^*=\min_{x\in X}\norm{x-q}_2$ and let $\mathrm{diam}(X)=\max_{i\ne j}\norm{x_i-x_j}_2$. If $R^*\ge R_{\max}$ for $R_{\max}=\mathrm{diam}(X)/(c-1)$, then every $x\in X$ satisfies $\norm{q-x}_2\le\norm{q-x^*}_2+\norm{x^*-x}_2\le R^*+\mathrm{diam}(X)\le cR^*$. So every data point, including the default output of generalized greedy search, is a $c$-approximate nearest neighbor. We therefore only need to handle queries with $R^*<R_{\max}$. All such queries satisfy $\norm{q}_2\le Q$ for $Q=\max_{x\in X}\norm{x}_2+R_{\max}$. We will use $Q$ as the query norm bound in \Cref{lem:one-table-simulation-main} and \Cref{thm:fixed-radius-steiner}.

With this observation in place, we use a standard reduction from \cite{IM98} between fixed-radius and $c$-approximate search. We build fixed-radius data structures for a geometrically increasing sequence of radii $R_0<R_1<\cdots<R_m$, where $R_{i}=\gamma^i R_0$ for some constant $\gamma>1$. We return the point found by the smallest-radius data structure that does not return FAIL. The smallest radius with $R_i\ge R^*$ satisfies $R_i<\gamma R^*$. With probability at least $1-\delta$, the structure at this radius does not return FAIL. So the returned point comes from this radius or a smaller one, and is within distance $cR_i<\gamma c R^*$ of $q$. 

Since we only need to handle queries with $R^*<R_{\max}$, the largest radius, $R_m$, only needs to be at least $R_{\max}$. For the smallest radius, let $r=\min_{i\ne j}\norm{x_i-x_j}_2$ and set $R_0=r/(c+1)$. Radii smaller than $R_0$ are not needed. To see why, consider a query that is very close to a data point, with $R^*< R_0$. Every other data point is at distance at least $r-R^*> r-R_0 = cR_0$ from $q$, so the structure at radius $R_0$ must return $q$'s true nearest neighbor. Recalling that we let $\Delta=\mathrm{diam}(X)/r$ denote the aspect ratio of the dataset, we have $R_{\max}/R_0=2(c+1)\Delta/(c-1)$. So the number of radii is $m+1=O(\log_{\gamma}(R_{\max}/R_0))=O(\log\Delta)$.

To implement the multi-level procedure with generalized greedy search, we do not need to identify the smallest radius that does not return FAIL. We simply chain the graphs for all radii together, in the same way that we chained tables in \Cref{thm:fixed-radius-steiner}. The successful terminal nodes and the failure node of the last table at radius $R_i$ have out-edges to the root of the first table at radius $R_{i+1}$, and Steiner points at smaller radii are placed so that, for every query with $\norm{q}_2\le Q$, they are farther from $q$ than all Steiner points at larger radii. A single run of generalized greedy search pops the bucket representatives from every table at every radius. Since the search outputs the closest data point  popped, its output is at least as good as the point found at the smallest radius. 
% Only the success event at that one radius is needed. This radius is determined by $q$ and $X$, not by the random hash functions, so the failure probability stays at most $\delta$ with no union bound over radii. 
The graph size and query time grow by a factor of $O(\log \Delta)$, the number of radii.

One small issue still remains: the reduction to fixed-radius search costs a factor of $\gamma$ in the approximation guarantee. To absorb this loss, we build each fixed-radius structure with a slightly smaller approximation factor $c_0\in(1,c)$ and set $\gamma=c/c_0$, so that our final approximation factor is $\gamma c_0 = c$, as desired. Naively, we might think this would increase the number of tables, and thus the runtime of search. However, from \Cref{thm:datar-lshSimple} we have that the number of tables required by \cite{IM98} is $O(n^\rho \log(1/\delta))$ where $\rho$ is a constant depending on $c$ and strictly less than $\frac{1}{c-1}$. Because of this strict inequality, we can pick $c_0$ close enough to $c$, depending only on $c$, so that $\rho$ at approximation factor $c_0$ is still below $\frac{1}{c-1}$. Each structure therefore still uses $L=O(n^{\frac{1}{c-1}}\log(1/\delta))$ tables, and the runtime claims of \Cref{thm:ubIntro} hold as stated.

\section{Empirical Evaluation}
\label{sec:experiments}

We conclude with an initial exploration into the empirical implications of our theoretical results. In particular, we study how augmenting existing graph-based search methods via the addition of Steiner points into the search graph can improve query time vs. accuracy for nearest neighbor search.

\smallskip

\noindent
\textbf{Experimental Set Up.}
We evaluate a very simple approach: we construct Steiner points as $k$-means cluster centroids of the dataset vectors, with $k \in \{0.05n,\, 0.15n,\, 0.25n\}$. The centroids are appended to the original dataset and we build a Vamana (DiskANN) search graph \citep{SubramanyaDevvritKadekodi:2019} over the augmented point set using ParlayANN's \citep{ManoharShenBlelloch:2024} open source implementation. We use degree bound $R = 32$, pruning parameter $\alpha = 1.2$, and build beam width $L = 128$. We give additional results with $R = 64$ in \Cref{app:experiments}. Crucially, the Steiner points are added as nodes in the search graph, but are not returned in the output of greedy search. We compare the performance of these search graphs with that of a baseline graph, where no Steiner points are included.

In particular, we run beam search with varying beam widths on the augmented and baseline graphs, in order to produce a runtime vs. accuracy curve for each graph. We measure accuracy as Recall@$k$ for $k \in \{1,10,100\}$ against ground truth $k$-nearest neighbors computed via brute-force. Recall@$k$ is the fraction of the true $k$-nearest neighbors appearing in the algorithm's top-$k$ output, averaged over queries. It is a standard accuracy metric for ANN benchmarks \citep{AumullerBernhardssonFaithfull:2020,SimhadriAumullerDouze:2026}. For each beam width, we plot recall against the average number of total distance computations per query, which is a common proxy for the query runtime.
Full implementation details are given in \Cref{app:experiments}.

\smallskip

\noindent
\textbf{Results.} Accuracy vs. runtime curves for the GloVe-200 and NYTimes datasets \citep{AumullerBernhardssonFaithfull:2020} using cosine similarity for the distance metric are shown in \Cref{fig:main-results}. %Results on additional datasets and in additional settings are given in \Cref{app:experiments}. 
On GloVe-200, Steiner-augmented graphs outperform the baseline, except at Recall@100. E.g., the search graph with $.05n$ Steiner nodes 
achieves the same Recall@1 as the baseline using roughly $25\%$ fewer distance
computations. We see similar gains for NYTimes -- here the Steiner augmented graphs consistently outperform, even at Recall@100. In \Cref{app:experiments} we give results on further datasets and in further settings. 
We frequently, although not always, see improvements with Steiner augmented graphs. Nevertheless, our experiments suggest that further study of how Steiner points can be effectively incorporated into graph-based search is warranted. Even the very simple $k$-means-based approach we consider can sometimes lead to significant performance gains. 
While we do not necessarily recommend adopting this approach in general, it represents a promising initial baseline for Steiner point construction.

\begin{figure*}[h]
    \centering
    \begin{subfigure}[t]{\textwidth}
        \centering
        \includegraphics[width=\linewidth]{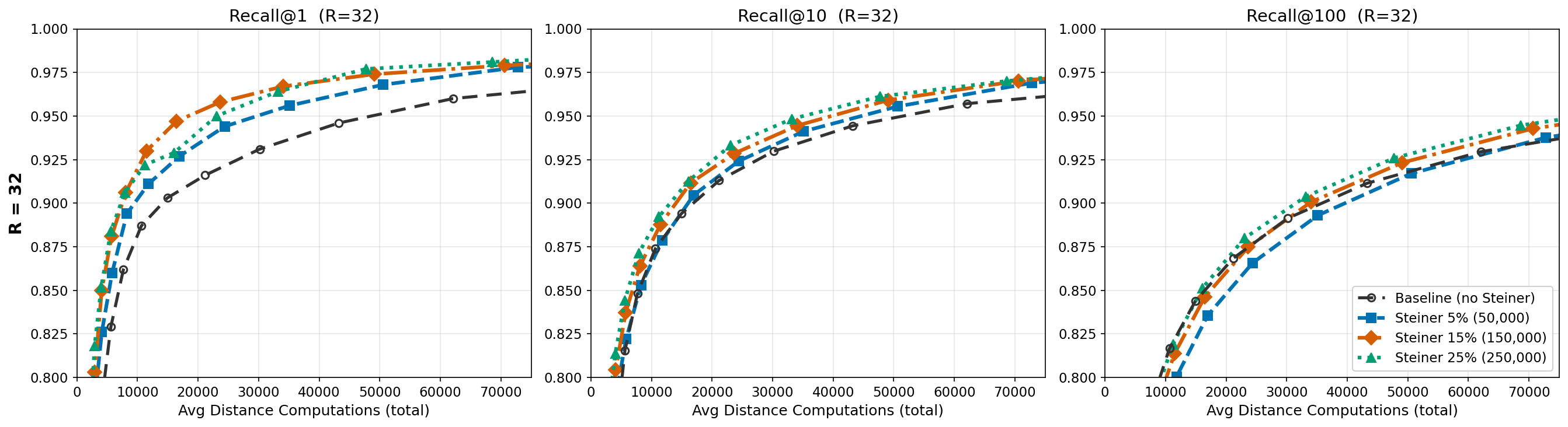}
        \caption{GloVe-200}
        \label{fig:glove200}
    \end{subfigure}
    \vspace{0.5em}
    \begin{subfigure}[t]{\textwidth}
        \centering
        \includegraphics[width=\linewidth]{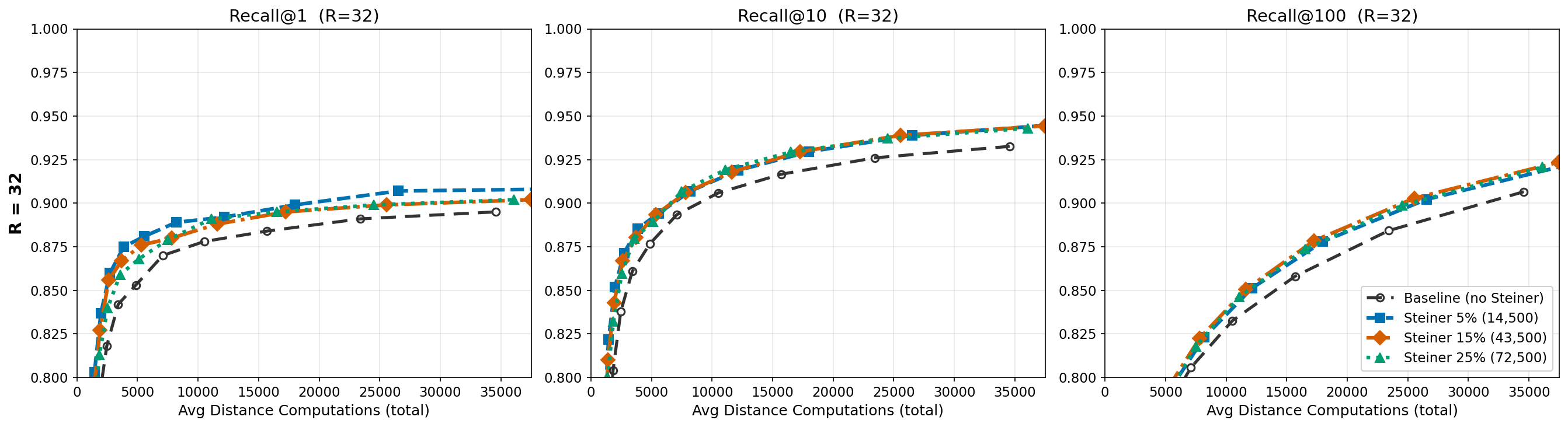}
        \caption{NYTimes}
        \label{fig:nytimes}
    \end{subfigure}
    \caption{Comparison of standard search graphs with Steiner augmented graphs in terms of recall vs.\ total distance computations, for nearest neighbor search under cosine similarity.}
    \label{fig:main-results}
\end{figure*}

%\section{Conclusion}

%\clearpage
\bibliographystyle{plainnat}
\bibliography{references}

%%%%%%%%%%%%%%%%%%%%%%%%%%%%%%%%%%%%%%%%%%%%%%%%%%%%%%%%%%%%

\appendix

\section{Low-Dimensional Lower Bound}
\label{sec:low-dim-single-source-lb}

We now prove \Cref{thm:lbLow}, which shows that, at least when using pure greedy search, graph-based methods with no Steiner points cannot match the sublinear query time guarantees of LSH, even in $O(\log n)$ dimensional Euclidean space. We restate a more complete version of the claim below.

\begin{theorem}[Theorem~\ref{thm:lbLow} restated]
\label{thm:single-source-low-dim-lb}
Let \(x_1,\ldots,x_n\) be drawn independently and uniformly from $\{-1,1\}^d$. With high probability, there is no search graph $G$ and starting node $s$ for $X$ for which pure greedy search successfully finds all in-dataset queries $q\in X$ in
\[
    o(dn^{1/2-\gamma}) \text{ time,}
\]
where,
$\gamma
    =
    \max\left\{
        \frac{\log\log n}{\log n},
        C\sqrt{\frac{\log n}{d}}
    \right\}
$
for a universal constant \(C>0\).
\end{theorem}
Note that, if $d \ge \frac{C^2}{\beta^2}\log n$ for any fixed constant $\beta > 0$, then
$
    C\sqrt{\frac{\log n}{d}}
    \le
     {\beta}.
$
Also, for all sufficiently large \(n\),
$
    \frac{\log\log n}{\log n}\le  {\beta} .
$
Therefore \(\gamma\le \beta\), so Theorem~\ref{thm:single-source-low-dim-lb} implies a runtime lower bound of $o(dn^{1/2-\beta})$ for any fixed constant $\beta$ when $d = O(\log n)$. This is how the simplified version of the result is stated in \Cref{thm:lbLow}.

Our proof of \Cref{thm:single-source-low-dim-lb} is based on a recent lower bound of \cite{DiwanGouMusco:2024} on the minimum average degree required to construct a ``navigable graph'' for $X$. While navigability of $G$ is a sufficient condition for pure greedy search to find any query $q \in X$, it is not a necessary condition. As such, our lower bound can be viewed as strengthening the result from \cite{DiwanGouMusco:2024}.

To prove Theorem~\ref{thm:single-source-low-dim-lb}, we define the 
``fixed-radius near-neighborhood'', $\mathcal{N}(j)$, of a point, $x_j$, as in \cite{DiwanGouMusco:2024}. For a threshold
\(\tau>0\) to be determined later, define
\begin{align}\label{eq:fixedRadius}
    \mathcal{N}(j)
    :=
    \{i\in X : \langle x_i,x_j\rangle \ge \tau\}.
\end{align}

One can argue that for a random dataset, the near-neighborhoods $\mathcal N(j)$ behave roughly like random sets: for an appropriate setting of $\tau$, we can ensure that each has size $\Theta(\sqrt{n})$, but the overlap between any two neighborhoods is small.
\begin{lemma}[Random hypercube neighborhood bounds]
\label{lem:random-hypercube-neighborhoods}
There are universal constants \(c_0,C_0,C_1,C_2>0\) and a threshold
\[
    \tau = c_h\sqrt{d\log n}
\]
with \(c_h\in[1/3,1]\), such that the following holds. Let
\(x_1,\ldots,x_n\) be drawn independently and uniformly from \(\{-1,1\}^d\),
where \(d\ge C_2\log n\). With high probability, the following three properties
hold simultaneously:
\begin{enumerate}
    \item 
    $x_1,\ldots,x_n\text{ are unique. I.e., }i \neq j \implies x_i \neq x_j.$
\item 
    $|\mathcal{N}(j)|\ge c_0\sqrt n
    \qquad\text{for every }j\in[n].$
\item $
    |\mathcal{N}(i)\cap \mathcal{N}(j)|
    \le
    C_0\max\left\{\log n,n^{C_1\sqrt{\log n/d}}\right\}
    \qquad\text{for all distinct }i,j\in[n].$
    \end{enumerate}
\end{lemma}

\begin{proof}
The first claim follows from a simple union bound -- the points are unique except with probability ${n \choose 2}/2^d = 1/\poly(n)$ when $d > C_2 \log n$ for sufficiently large $C_2$. The next two claims follow from the random hypercube neighborhood analysis of
~\cite{DiwanGouMusco:2024}. Their fixed-radius neighborhood is
defined as
\[
    O_j=\{x_i:\langle x_i,x_j\rangle\ge c_h\sqrt{d\log n}\},
\]
which is exactly our \(\mathcal{N}(j)\) with
\(\tau=c_h\sqrt{d\log n}\). Their Claim~8 shows that, for this threshold,
\[
    |O_j|\ge \sqrt n/6
\]
simultaneously for all \(j\in[n]\), and their Claim~9 shows that
\[
    |O_i\cap O_j|
    \le
    10\max\{\log n,n^{c\sqrt{\log n/d}}\}
\]
simultaneously for all distinct \(i,j\in[n]\), for a universal constant \(c\).

Their statements are written with constant success probability. However, the
proof of Claim~8 gives failure probability \(\exp(-\Omega(\sqrt n))\) for each
fixed \(j\), and then union bounds over all \(j\). The proof of Claim~9 gives
polynomially small failure probability for each fixed pair \(i,j\), and then
union bounds over all \(O(n^2)\) pairs. Therefore, by increasing the absolute
constants in the dimension lower bound and in the intersection bound, the same
argument gives the high-probability version stated here.
\end{proof}

We next argue that, in order for greedy search to succeed on all in-dataset queries, the source node $s$ needs an out-neighbor within the near-neighborhood of \emph{all its own near-neighbors}. Using the bounds of \Cref{lem:random-hypercube-neighborhoods}, this allows us to show that $s$ needs degree $\Omega(n^{1/2-\gamma})$, and thus prove the runtime bound of \Cref{thm:single-source-low-dim-lb}.
\begin{lemma}[First-hop covering constraint]
\label{lem:first-hop-cover}
Consider a dataset $X$ whose points are unique -- i.e., $i \neq j \implies x_i \neq x_j$. 
Fix a directed search graph $G$ and starting node $s$, for which pure greedy search on $G$ starting from $s$ succeeds for all queries  $q \in X$. Then,  letting $N_G^+(s)$ denote the out-neighborhood of $s$ in $G$, %there exists
%\(v\in N_G^+(s)\) such that
% \[
%     x_v\in \mathcal{N}(y).
% \]
% Equivalently,
\[
    N_G^+(s)\cap \mathcal{N}(y)\neq \emptyset
    \qquad\text{for every }y\in \mathcal{N}(s)\setminus\{s\}.
\]
\end{lemma}

\begin{proof}
Fix \(y\in \mathcal{N}(s)\setminus\{s\}\). By the definition of $\mathcal{N}(s)$ in \eqref{eq:fixedRadius}, we have 
$\langle x_y,x_s\rangle \ge \tau.$
Since the points in $X$ are unique, greedy search in \(G\) starting from \(s\) must eventually reach \(x_y\) on query $x_y$. In particular, it cannot stop at the
initial vertex \(x_s \neq x_y\). Hence \(x_s\) must have an out-neighbor
\(v\in N_G^+(s)\) satisfying
\[
    \|x_v-x_y\|_2 < \|x_s-x_y\|_2.
\]
Otherwise the algorithm will return \(x_s\).
Using
$
    \|a-b\|_2^2 = 2d-2\langle a,b\rangle
$
for \(a, b\in\{-1,1\}^d\), this inequality is equivalent to
$    \langle x_v,x_y\rangle > \langle x_s,x_y\rangle.$
Combining this with \(\langle x_s,x_y\rangle\ge \tau\), we get
$\langle x_v,x_y\rangle > \tau.$
Therefore \(x_v\in \mathcal{N}(y)\), and so  $N_G^+(s)\cap \mathcal{N}(y)\neq \emptyset$ as claimed.
\end{proof}

We now combine \Cref{lem:random-hypercube-neighborhoods,lem:first-hop-cover} to argue that the starting vertex must have high degree. This will immediately imply the runtime lower bound of \Cref{thm:single-source-low-dim-lb}, since greedy search at least needs to compute the distance from the query $q$ to any neighbor of the starting vertex.

\begin{lemma}[Source degree lower bound]
\label{lem:source-degree-lb}
Let $X$ be as in Lemma~\ref{lem:random-hypercube-neighborhoods}, and assume the three events in the lemma hold.
Fix a directed search graph $G$ and starting node $s$, for which pure greedy search on $G$ starting from $s$ succeeds for all queries  $q \in X$. Then the out-degree of $s$ in $G$ satisfies:

\[
    \deg_G^+(s)
    =
    \Omega(n^{1/2-\gamma}),\quad \text{ where }    \gamma
    =
    \max\left\{
        \frac{\log\log n}{\log n},
        C_1\sqrt{\frac{\log n}{d}}
    \right\}.
\]
\end{lemma}

\begin{proof}
We can apply Lemma \Cref{lem:first-hop-cover} to $X$ since the first claim of\Cref{lem:random-hypercube-neighborhoods} ensures that the points in $X$ are unique.
By Lemma~\ref{lem:first-hop-cover}, every point
\(y \in \mathcal{N}(s)\setminus\{s\}\) must be ``covered'' by at least one
out-neighbor of \(s\). I.e., we must have some \(v\in N_G^+(s)\) with  \(v\in \mathcal{N}(y)\) to ensure that $N_G^+(s)\cap \mathcal{N}(y)\neq \emptyset.$

We upper bound how many points in  \(\mathcal{N}(s)\setminus\{s\}\) a single
out-neighbor can cover. Fix \(v\in N_G^+(s)\). Any $y$ covered by \(v\) has $v \in \mathcal{N}(y)$ -- i.e., $\langle x_v, x_y \rangle \ge \tau$ by the definition of $\mathcal{N}(y)$ in \eqref{eq:fixedRadius}. By symmetry of the inner product, we thus also have $y \in \mathcal{N}(v)$.
Therefore, we have $y \in |\mathcal{N}(s)\cap \mathcal{N}(v)|.$

Since \(v\neq s\), the third claim of Lemma~\ref{lem:random-hypercube-neighborhoods}
gives
\[
    |\mathcal{N}(s)\cap \mathcal{N}(v)|
    \le
    C_0\max\left\{\log n,n^{C_1\sqrt{\log n/d}}\right\}.
\]
I.e., any \(v\in N_G^+(s)\) covers at most $C_0\max\left\{\log n,n^{C_1\sqrt{\log n/d}}\right\}$ vertices in $\mathcal N(s) \setminus \{ s\}.$
On the other hand, the second claim of Lemma~\ref{lem:random-hypercube-neighborhoods} gives $|\mathcal{N}(s)\setminus\{s\}| \ge c_0\sqrt n-1.$
Thus, we must have:
\[
    \deg_G^+(s)\cdot
    C_0\max\left\{\log n,n^{C_1\sqrt{\log n/d}}\right\}
    \ge
    c_0\sqrt n-1.
\]
I.e.,
\[
    \deg_G^+(s)
    =
    \Omega\left(
        \frac{\sqrt n}
        {\max\{\log n,n^{C_1\sqrt{\log n/d}}\}}
    \right) = \Omega(n^{1/2-\gamma})
\]
for $
    \gamma
    =
    \max\left\{
        \frac{\log\log n}{\log n},
        C_1\sqrt{\frac{\log n}{d}}
    \right\}.$
\end{proof}

\begin{proof}[Proof of Theorem~\ref{thm:single-source-low-dim-lb}]
Drawing \(x_1,\ldots,x_n\) independently and uniformly from \(\{-1,1\}^d\), with high probability, all three claims of Lemma~\ref{lem:random-hypercube-neighborhoods} hold. Thus, with high probability, by \Cref{lem:source-degree-lb}, for any search graph $G$ with starting node $s$ for which pure greedy search succeeds on all in-dataset queries, we have $\deg_G^+(s) =  \Omega(n^{1/2-\gamma}).$ I.e., greedy search must at least compute the distance from $q$ to $\Omega(n^{1/2-\gamma})$ points in the first step of the algorithm, requiring $\Omega(dn^{1/2-\gamma})$ time (even to read those datapoints).
\end{proof}

\section{Omitted Proofs for LSH Simulation with Steiner Points}
\label{app:steiner}

We now give complete proofs of the results in \Cref{sec:steiner}, beginning with \Cref{thm:datar-lshSimple}.

\subsection{Fixed Radius LSH Guarantee}
\begin{proof}[Proof of \Cref{thm:datar-lshSimple}]
Suppose an LSH function collides with probability at least $p_1$ for points within distance $R$ of $q$, and at most $p_2$ for points farther than $cR$. The LSH scheme of \cite{IM98} uses $k=O(\log n/\log(p_1/p_2))$ hash functions per table and $L=O(n^\rho\log(1/\delta))$ tables, with $\rho=\frac{\log(1/p_1)}{\log(p_1/p_2)}$. Our main challenge is thus to argue that, for a suitable bucket width, the hash function of \cite{DIIM04} has $\rho< \frac{1}{c-1}$, with $p_1$ and $p_2$ depending only on $c$. 

\cite{DIIM04} analyze the closely related quantity $\rho'=\frac{\log(1/p_1)}{\log(1/p_2)}$. $\rho=\frac{\rho'}{1-\rho'}$, so $\rho<1/(c-1)$ if and only if $\rho'<1/c$. They show analytically that $\rho'\le(1+\epsilon)/c$ for any $\epsilon>0$ and a sufficiently large bucket width, and numerically that $\rho'<1/c$ for $c\in(1,10]$. We give a short proof that $\rho'<1/c$, and thus $\rho<1/(c-1)$, for every $c>1$. Combined with \cite{IM98}, this proves the claim. 
% The strictness of the inequality is also used in the reduction from ANN to fixed-radius search.

\label{app:steiner-lsh-guarantee}
For search radius $R$, write the bucket width as $w=WR$, and let $Z$ be a standard Gaussian random variable. As shown in \cite{DIIM04}, for two points at distance $tR$, a single hash function collides with probability
%w = WR = r. tR = c. int_0^WR 1/(tR) f(|Z|/tR) (1-t/(WR))
%0 when |Z| > W/t.

\[
    p(t)=\E\left[\max\left\{1-\frac{t|Z|}{W},0\right\}\right],
\]
which decreases with $t$. Write $\alpha=\E|Z|=\sqrt{2/\pi}$. Without the truncation at zero, we would have $p(t) = 1-\alpha t/W$. The truncation changes this expectation only when $|Z|>W/t$. Thus
\begin{align}\label{eq:prob1}
    p(t)=1-\frac{\alpha t}{W}
       +\E\left[\left(\frac{t|Z|}{W}-1\right)
           \mathbf{1}_{\{|Z|>W/t\}}\right].
\end{align}
We have 
\begin{align*}
0 \le E\left[\left(\frac{t|Z|}{W}-1\right)
           \mathbf{1}_{\{|Z|>W/t\}}\right] &\le \frac{t}{W} E\left[|Z| \cdot 
           \mathbf{1}_{\{|Z|>W/t\}}\right]\\
           &= \frac{2t}{W\sqrt{2\pi}}\int_{W/t}^{\infty}z e^{-z^2/2}\,dz
      =\frac{\alpha t}{W}e^{-W^2/(2t^2)}.
\end{align*}
For every fixed $t>0$, this error decays faster than any inverse power of $W$. E.g.,, combined with \eqref{eq:prob1}, we have
\[
    p(t)=1-\frac{\alpha t}{W}+o(W^{-3}).
\]
We next use the Taylor expansion $\log(1/(1-u))=u+u^2/2+O(u^3)$ as $u\to0$, taking $u=1-p(t)=\alpha t/W+o(W^{-3})$. This gives that for large enough $W$,
\[
    \log\frac{1}{p(t)}=\frac{\alpha t}{W}
       +\frac{\alpha^2t^2}{2W^2}+O(W^{-3}).
\]
Finally, we let $p_1=p(1)$ and $p_2=p(c)$, and  apply the expansion at these points. This gives:
\[
\begin{aligned}
    \rho'=\frac{\log(1/p_1)}{\log(1/p_2)}
    &=\frac{\alpha/W+\alpha^2/(2W^2)+O(W^{-3})}
            {\alpha c/W+\alpha^2c^2/(2W^2)+O(W^{-3})} \\
    &=\frac{1}{c}\cdot
        \frac{1+\alpha/(2W)+O(W^{-2})}
             {1+\alpha c/(2W)+O(W^{-2})} \\
    &=\frac{1}{c}\left(1-\frac{\alpha(c-1)}{2W}+O(W^{-2})\right).
\end{aligned}
\]
The last equality uses $(1+u)^{-1}=1-u+O(u^2)$. For fixed $c>1$, the negative term $-\alpha(c-1)/(2W)$ dominates the $O(W^{-2})$ remainder when $W$ is sufficiently large with respect to $c$. So, there is a choice of $W$, depending only on $c$, for which $\rho'<1/c$ and therefore $\rho=\rho'/(1-\rho')<1/(c-1)$, as desired. Since $p(t)$ is continuous, $\rho = \log(1/p(1))/\log(p(1)/p(c))$ is continuous in $c$; thus we can choose $c_0\in(1,c)$ sufficiently close to $c$ such that 
\begin{equation}
    \label{eq:lsh-reduction-slack}
    \frac{\log(1/p(1))}{\log(p(1)/p(c_0))}<\frac{1}{c-1},
\end{equation}
as required for the reduction to fixed radius search describe in \Cref{sec:fixedReduction}.

We conclude by noting that the number of hash functions $k=O(\log n/\log(p_1/p_2)) = O(\log n)$, since $\log(p_1/p_2)$ is a constant depending only on $c$.
\end{proof}

\subsection{The Steiner Comparison Gadget}
\label{app:steiner-comparison-gadget}

We next prove \Cref{lem:steiner-comparison-gadget}, the basic building block of our routing trees: two Steiner points whose comparison by greedy search evaluates a linear threshold function.

\begin{replemma}{lem:steiner-comparison-gadget}
Fix a nonzero vector $a\in\R^d$ and a threshold $\theta\in\R$. There exist two
Steiner points $s_L,s_R\in\R^d$ such that, for every query $q\in\R^d$,
\[
    \|q-s_L\|_2\, < \|q-s_R\|_2
    \quad \Longleftrightarrow \quad
    \langle a,q\rangle < \theta .
\]
Further, we can pick $s_L$ and $s_R$ so that $\norm{s_L}_2$ and $\norm{s_R}_2$ are arbitrarily large.
\end{replemma}
\begin{proof}
Let $\bar a=a/\norm{a}_2$ and $\tau=\theta/\norm{a}_2$. For any $\lambda>0$, set $s_L=(\tau-\lambda)\bar a$ and $s_R=(\tau+\lambda)\bar a$. Then
\[
    \norm{q-s_L}_2^2-\norm{q-s_R}_2^2
    =4\lambda(\langle\bar a,q\rangle-\tau)
    =\frac{4\lambda}{\norm{a}_2}(\langle a,q\rangle-\theta).
\]
This proves the desired equivalence. Moreover, both $\norm{s_L}_2$ and $\norm{s_R}_2$ are at least $\lambda-|\tau|$, and so can be made arbitrarily large.
\end{proof}

\subsection{Simulating One Hash Table}
\label{app:one-table-simulation}

We now use the comparison gadget of \Cref{lem:steiner-comparison-gadget} to prove our simulation result for a single LSH table, \Cref{lem:one-table-simulation-main}.
\begin{replemma}{lem:one-table-simulation-main}
Fix any \(0\le Q<\infty\), concatenated \cite{DIIM04} hash function  $g(x)=(h_{a_1,b_1}(x),\ldots,h_{a_k,b_k}(x))$,
and a dataset \(X\subset\mathbb{R}^d\) of size \(n\).
There is a directed graph \(G\) over \(X\) and $O(kn)$ additional Steiner points, with $O(kn)$ edges, and a Steiner starting node
\(s\), with the following
property.  For every query \(q\) with $\norm{q}_2 \le Q$, generalized greedy search
initiated at \(s\) and run for $T=O(k\log n)$ steps either:
\begin{enumerate}
    \item terminates at a designated Steiner failure node $s_{F}$, if \(X_{g(q)}=\varnothing\); or
    \item reaches a designated terminal node $s_{g(q)}$ if
          \(X_{g(q)}\neq\varnothing\), with a single out-edge to an arbitrary point in \(X_{g(q)}\), and then pops this out-neighbor off the search queue ($\mathcal{C}$ in \Cref{alg:gen_greedy_search}).
\end{enumerate}
All nodes in $G$ have out-degree $\le 2$. The search uses $O(T)$ distance computations and $\tilde O(Td)$ total time.
\end{replemma}

\begin{proof}
We first specify the tree and its edges, and then describe how to place its Steiner nodes, to ensure that generalized greedy search correctly progresses down the tree. For a nonempty set $Y\subseteq X$ and one hash function $h_{a_i,b_i}$, consider the occupied bucket values $\{h_{a_i,b_i}(x):x\in Y\}$. Each occupied value $j$ corresponds to values of $\frac{\langle a_i, x \rangle + b_i}{w}$ lying in the interval $[j,j+1]$. These intervals, together with the nonempty gaps between them and the two unbounded outer intervals, partition the real line into at most $2|Y|+1$ intervals. We build a balanced binary search tree over these intervals. Each internal node compares $\langle a_i,q\rangle$ with a boundary $\theta$, and its two children implement that comparison via \Cref{lem:steiner-comparison-gadget}. A leaf represents either an occupied bucket or a failure interval. The tree has $O(|Y|)$ nodes and depth $O(\log(|Y|+1))$.

For $k>1$, start with $Y=X$ and the first hash coordinate. Each successful leaf becomes the root of a tree for the next coordinate, restricted to the points that agree with \emph{all} previously selected hash values. I.e.,, for a prefix $I=(i_1,\ldots,i_{\ell})$, its active set is $X_I=\{x\in X:h_{a_1,b_1}(x)=i_1,\ldots,h_{a_{\ell},b_{\ell}}(x)=i_{\ell}\}.$

After all $k$ coordinates, each leaf corresponding to a non-empty interval is the terminal node $s_I$ for a nonempty bucket $X_I=\{x\in X:g(x)=I\}$. Give $s_I$ one out-edge to an arbitrary fixed point of $X_I$. Give every failure leaf one out-edge to the shared Steiner failure node $s_F$. Dataset nodes and $s_F$ have no out-edges.

At each coordinate, the active sets partition $X$, so the sum of their sizes is $n$. Hence the complete graph has $O(kn)$ Steiner nodes and edges. Every root-to-leaf path has length $O(k\log n)$, and every node has out-degree at most two. 

We next choose the Steiner points so that generalized greedy search follows the root-to-leaf path determined by $g(q)$. As discussed, we make  deeper Steiner nodes, and all data points, closer to every query with $\norm{q}_2\le Q$ than shallower Steiner nodes. Then the search can never return to an unchosen sibling. Let $M=\max_{x\in X}\norm{x}_2$. Pick $s_F$ to have norm greater than $M+2Q$. Assign $s_F$ a depth greater than every tree node, and place the tree nodes in decreasing order of depth. When placing a pair of children at one depth, use the free parameter $\lambda$ in \Cref{lem:steiner-comparison-gadget} to make both norms exceed the largest norm at any greater depth by more than $2Q$. The root, whose location is unconstrained, is placed last with the same norm separation.

For any $\norm{q}_2\le Q$ and Steiner nodes $u,v$ with $u$ at smaller depth than $v$, this placement ensures
\[
    \norm{q-u}_2\ge\norm{u}_2-Q
    >\norm{v}_2+Q\ge\norm{q-v}_2.
\]
Also, every data point $x$ is closer to $q$ than every Steiner point $u$, since $\norm{q-x}_2\le M+Q<\norm{u}_2-Q\le\norm{q-u}_2$. When an internal node is popped, its children are therefore closer than all earlier unchosen siblings still in $\mathcal C$. The next node popped is its correct child, and thus the greedy search follows the comparisons defining $g(q)$. If the bucket is empty, it reaches a failure leaf and then pops $s_F$. Otherwise, it reaches $s_{g(q)}$, inserts its representative data point, and pops that data point next. Stopping at $s_F$ or at this data point proves the routing guarantee.

The search takes $T=O(k\log n)$ steps. Each step performs at most two distance computations and $O(1)$ priority-queue and dictionary operations on structures of size $O(T)$, giving $\tilde O(T d)$ query time. %For later use, computing the hashes and sorting the occupied bucket values takes $O(dkn+kn\log n)$ preprocessing time. Once the trees are built, placing their nodes from the leaves upward takes $O(dkn)$ additional time.
\end{proof}

\subsection{Fixed-Radius Upper Bound}
\label{app:all-table-simulation}

We next describe how to chain $L$ one-table graphs together to simulate the full fixed-radius search procedure of \Cref{thm:datar-lshSimple}, proving \Cref{thm:fixed-radius-steiner}.

\begin{reptheorem}{thm:fixed-radius-steiner}
Fix $R>0$, $c>1$, $\delta\in(0,1)$, and any $Q<\infty$. There is a constant $\rho$ depending only on $c$ with $\rho < \frac{1}{c-1}$ such that, for any $n$-point dataset $X\subseteq\R^d$, there is a randomized polynomial time algorithm constructing a directed graph $G$ on $X$ and $\tilde O(n^{1+\rho}\log(1/\delta))$ additional Steiner nodes, and a Steiner starting node $s$, with the following property. For every fixed $q$ with $\norm{q}_2\le Q$ and $\min_{x\in X}\norm{x-q}_2\le R$, generalized greedy search initialized at $s$ pops at least one point $x\in X$ with $\norm{x-q}_2\le cR$ within $T=\tilde O(n^{\rho} \log(1/\delta))$ steps, with probability at least $1-\delta$. All nodes in $G$ have out-degree at most two, so running generalized greedy search for $T$ steps requires $O(T)$ distance computations and takes $\tilde O(Td)$ time.
\end{reptheorem}
\begin{proof}
Choose $k$ and $L$ as in \Cref{thm:datar-lshSimple}, and build the $L$ one-table graphs from \Cref{lem:one-table-simulation-main}. These graphs all contain the dataset nodes, but have disjoint sets of Steiner nodes. Use the first table's root as the global starting node $s$. In every table except the last, add an out-edge from each terminal node $s_I$ and from $s_F$ to the next table's root. Thus a terminal node has at most two out-edges (to the corresponding dataset node and to the next root), and every other node still has out-degree at most two. Dataset nodes have no out-edges.

Apply the same Steiner node placement approach as in the proof of \Cref{lem:one-table-simulation-main}, processing the tables from last to first, so that Steiner nodes in the earlier tables have larger norms. In addition to the within-table separation, ensure that every Steiner node norm in an earlier table exceeds every Steiner node norm in a later table by at least $2Q$. All Steiner node norms still exceed $M+2Q$, where $M=\max_{x\in X}\norm{x}_2$. %As before, changing the free parameters of the comparison gadgets preserves their thresholds.

Fix $q$ with $\norm{q}_2\le Q$. Within each table, search routes as in \Cref{lem:one-table-simulation-main}. In a nonempty bucket, popping the terminal node inserts its representative from the dataset and the next table's root. The representative, if not seen before, is popped next, since data points are closer to $q$ than all Steiner points. The next table's root is closer than all remaining Steiner nodes from earlier tables, so the search then proceeds to that table. In an empty bucket, search pops $s_F$ and proceeds directly to the next root.

From \Cref{lem:one-table-simulation-main} it follows that, within $T=O(Lk\log n)$ steps, search has popped the representative of every nonempty queried bucket. In particular, \Cref{thm:datar-lshSimple} implies that a point within distance $cR$ is popped with probability at least $1-\delta$ whenever some point lies within distance $R$. %We therefore terminate search after $T$ steps, or earlier if the queue empties.

The graph has $O(Lkn) = \tilde O(n^{1+ \rho} \log(1/\delta))$ Steiner nodes and edges. %The preprocessing bound from \Cref{app:one-table-simulation} gives $O((d+\log n)Lkn)$ construction time. 
As in \Cref{lem:one-table-simulation-main}, each step requires $O(1)$ distance computations and so the overall query time is $\tilde O(Td)$. %Substituting $k=O(\log n)$ and $L=O(n^{1/(c-1)}\log(1/\delta))$ proves all stated bounds.
%The ANN reduction in \Cref{app:steiner-lsh-to-ann} does not need a FAIL output. To solve the $(R,c)$-near-neighbor problem itself, return the closest popped data point if it is within distance $cR$ of $q$, and FAIL otherwise. With $Q=M+cR$, FAIL is also correct for every query with $\norm{q}_2>Q$.
\end{proof}

\subsection{The reduction from ANN to fixed-radius search}
\label{app:steiner-lsh-to-ann}

Finally, we combine our fix-radius search result (\Cref{thm:fixed-radius-steiner}) with the standard reduction from c-approximate nearest neighbor search of \cite{IM98} to prove our main result on fast graph-based search with Steiner points, \Cref{thm:ubIntro}.

\begin{reptheorem}{thm:ubIntro}
There exists a polynomial time algorithm that, given any $n$-point dataset $X \in \R^d$, constructs a search graph $G$ over $X$ and $\tilde O(n^{1+1/(c-1)} \log(\Delta) \log(1/\delta))$ additional Steiner points, such that generalized greedy search on $G$ run for $\tilde O(n^{1/(c-1)}  \log(\Delta) \log(1/\delta))$ steps starting from  an appropriate start node, solves $c$-approximate Euclidean nearest-neighbor search in $\tilde O(d n^{1/(c-1)})$ time, with probability at least $1-\delta$. Here $\Delta = \frac{\max_{i \neq j} \norm{x_i-x_j}_2}{\min_{i \neq j} \norm{x_i-x_j}_2}$ is the aspect ratio of the dataset. The big-O notation hides constants that depend on $c$.
\end{reptheorem}
\begin{proof}
As described in \Cref{sec:fixedReduction}, we chain fixed-radius search graphs for geometrically spaced radii into a single search graph. As shown, letting $R_{\max}=\frac{\mathrm{diam}(X)}{c-1}$ and $Q=\max_{x\in X}\norm{x}_2+R_{\max}$, for any query $q$ with $\norm{q}_2 > Q$, all datapoints are $c$-approximate nearest neighbors. This includes any dataset point popped during the search, or the default output if no dataset point is every popped. Thus, correctness only needs to be shown for queries with $\norm{q_2} \le Q$. 

\smallskip

\noindent\textbf{Radii and Graph Construction.} Choose $c_0\in(1,c)$ satisfying \eqref{eq:lsh-reduction-slack}. Let $r=\min_{i\ne j}\norm{x_i-x_j}_2$ and build fixed-radius search graphs as in \Cref{thm:fixed-radius-steiner} with approximation factor $c_0$ for the radii
\[
    R_0=\frac{r}{c+1},
    \qquad
    R_i=(c/c_0)^iR_0\quad (i=0,\ldots,m),
    \qquad
    m=\left\lceil\log_{c/c_0}(R_{\max}/R_0)\right\rceil.
\]
Since $R_{\max}/R_0=(c+1)\Delta/(c-1)$, there are $O(\log\Delta)$ radii, with a constant depending on $c/c_0$.

Concatenate the search graphs for each radii into one graph using the same idea is in \Cref{thm:fixed-radius-steiner}: terminal Steiner nodes from the last table in one search graph are connected to the starting vertex for the next search graph, and Steiner node norms are chosen so that nodes in the earlier search graphs are always further from $q$ than nodes in the later search graphs. Data set nodes remain sinks with no out-edges. By \Cref{thm:fixed-radius-steiner}, for every query with $\norm{q}_2\le Q$, one run of generalized greedy search pops all selected representatives within $T=O((m+1)Lk\log n) = \tilde O(n^{1+1/(c-1)} \log(\Delta) \log(1/\delta))$ steps (and running time since all out-degrees are still at most two). We terminate search after $T$ steps, or earlier if the queue empties. The output is the closest data point popped, or the default output $x_1$ if no data point was popped.

\smallskip

\noindent\textbf{Correctness.}  Fix $q$, let $R^*=\min_{x\in X}\norm{x-q}_2$, and assume $R^*<R_{\max}$, so that $\norm{q}_2\le Q$. Let $i^*=0$ if $R^*\le R_0$. Otherwise, let $i^*$ be the smallest index with $R_{i^*}\ge R^*$, which exists since $R_m\ge R_{\max}>R^*$. We only use the fixed-radius guarantee at the single radius $R_{i^*}$. This index is determined by $q$ and $X$ alone, and $q$ is fixed independently of the random hash functions. So \Cref{thm:fixed-radius-steiner} applies at radius $R_{i^*}$: with probability at least $1-\delta$, the search pops a point $x$ with $\norm{q-x}_2\le c_0R_{i^*}$. %No union bound over radii is needed, so the overall failure probability is at most $\delta$.

On this event, $x$ is a valid answer. If $i^*=0$, then $x=x^*$, since every other $y\in X$ satisfies
\[
    \norm{q-y}_2\ge r-R^*\ge r-R_0 = c R_0 >c_0R_0.
\]
If $i^*>0$, then $R_{i^*-1}<R^*$, so $R_{i^*}<(c/c_0)R^*$ and
\[
    \norm{q-x}_2\le c_0R_{i^*}\le cR^*.
\]
Since generalized greedy search outputs the closest data point it popped, its output is only close to $q$ than $x$, and so is also a valid answer.

\end{proof}

\section{Empirical Evaluation: Details and Additional Results}\label{app:experiments} 

We now provide full details of our experimental setup (search procedure, datasets, infrastructure, and Steiner point generation) and additional recall-vs-computation plots for the datasets not shown in the main body. We include results for SIFT (Euclidean distance and cosine similarity), Fashion-MNIST, and YFCC, as well as results with the maximum degree $R=64$ for GloVe and NYTimes.

\smallskip

\noindent \textbf{Preliminary experiments.}
Prior to the large-scale evaluation, we conducted extensive preliminary experiments on smaller subsets ($12$K vectors) across multiple datasets and multiple strategies for constructing Steiner points. In addition to $k$-means cluster centroids, we evaluated midpoints of randomly sampled pairs of data points and midpoints of pairs that are close in the metric space but far apart in graph distance (i.e., pairs requiring many hops to reach each other despite being nearby in $\mathbb{R}^d$). Among all strategies tested, cluster centroids consistently yielded the best recall-vs-computation tradeoffs, which motivated our decision to focus on centroid-based Steiner augmentation for the full-scale experiments. These preliminary runs, along with several failed configurations due to metric mismatches and timeout issues, account for roughly twice the compute of the final reported experiments.

\subsection{Experimental Setup}

\textbf{Search procedure.}
We build the Vamana (DiskANN) \citep{SubramanyaDevvritKadekodi:2019} graph over the augmented point set using ParlayANN's~\citep{ManoharShenBlelloch:2024} parallel C++ implementation. At query time, we run beam search with varying beam widths $W$ on the augmented graph. Steiner nodes participate fully in the beam (they can be expanded and their edges traversed) but they are excluded from the final result set. We ran beam search with varying beam widths using either Euclidean distance or inverse cosine similarity as our distance metric, dependent on the dataset. For Steiner-augmented graphs, the total number of distance computations includes overhead of evaluating distances to Steiner nodes during beam search. By default, ParlayANN randomly permutes the dataset, hence Steiner points are at random positions in the order of insertion. Steiner points were generated using FAISS's $k$-means implementation.

\smallskip

\noindent
\textbf{Parameters.} We evaluate each configuration across a geometric progression of beam widths $W = 4,\, 8,\, 16,\, 32, \dots$, doubling initially and later scaling by a factor of $1.5$ until the total number of distance computations reaches a fraction $0.1$ of the dataset size. We vary the maximum degree as $R \in \{32, 64\}$, and set build beam width $L = 128$, and pruning parameter $\alpha = 1.2$ across all experiments. For querying, we subsample 1000 queries uniformly at random from the test datasets for all experiments.

\smallskip

\noindent
\textbf{Dataset sources and preprocessing.}
GloVe-200, NYTimes, SIFT, and Fashion-MNIST are sourced from ANN-Benchmarks~\citep{AumullerBernhardssonFaithfull:2020}. YFCC is sourced from the NeurIPS'21 Big-ANN benchmark~\citep{SimhadriAumullerDouze:2026}; we uniformly subsample $1$M vectors from the full $10$M-vector dataset. 
We summarize all datasets in \Cref{tab:datasets}.

\begin{table}[t]
\centering
\begin{tabular}{|l|c|c|c|l|}
\hline
\textbf{Dataset} & \textbf{Dimensions} & \textbf{\# Data Points} & \textbf{\# Query Points} & \textbf{Metric} \\
\hline
GloVe-200 & 200 & 1,000,000 & 1,000 & Cosine \\
\hline
SIFT & 128 & 1,000,000 & 1,000 & Euclidean \\
\hline
SIFT & 128 & 1,000,000 & 1,000 & Cosine \\
\hline
NYTimes & 256 & 290,000 & 1,000 & Cosine \\
\hline
Fashion-MNIST & 784 & 60,000 & 1,000 & Cosine \\
\hline
YFCC & 192 & 1,000,000 & 1,000 & Cosine \\
\hline
\end{tabular}
\caption{Datasets used for evaluation.}
\label{tab:datasets}
\end{table}

\smallskip

\noindent
\textbf{Licenses.}
All software and datasets used in this work are available in the public domain. ParlayANN~\citep{ManoharShenBlelloch:2024}, used for graph construction, is released under the MIT License, as is FAISS, which we use for $k$-means clustering. Among the datasets, GloVe is released under the Public Domain Dedication and License (PDDL v1.0); SIFT is released under Creative Commons (CC) BY 4.0; Fashion-MNIST is released under the MIT License; and NYTimes is distributed via ANN-Benchmarks~\citep{AumullerBernhardssonFaithfull:2020} and is licensed under CC BY 4.0. The YFCC embeddings are sourced from the NeurIPS'21 Big-ANN benchmark~\citep{SimhadriAumullerDouze:2026}, licensed under CC BY 4.0.

\smallskip

\noindent
\textbf{Infrastructure.}
All large-scale experiments were run on the Unity high-performance computing cluster at the University of Massachusetts Amherst. Each job was allocated a single node with $64$ CPU cores and $64$\,GB of RAM on the \texttt{cpu} partition, running Ubuntu 24.04 with Python 3.11.7. Graph construction was performed using ParlayANN, a parallel C++ implementation of the Vamana algorithm that leverages all available cores.

\smallskip

\noindent
\textbf{Runtime.}
For the $1$M-vector datasets (GloVe, SIFT, YFCC), graph construction and search evaluation across all beam widths for $1{,}000$ queries took $6$--$7$ hours per configuration. For the smaller datasets (NYTimes at $290$K, Fashion-MNIST at $60$K), each configuration completed in $1$--$2$ hours. Each dataset required $6$ SLURM jobs (three Steiner percentages $\times$ two degree bounds), bringing the total compute for the full experimental sweep to approximately $300$ CPU-hours.

\smallskip

\noindent
\textbf{Ground truth.}
Ground truth nearest neighbors for computing recall values were computed by brute-force search. Since this was a single-time computation, the total computation time required was insignificant compared to the time required to run the full experimental setup.

%For datasets where the search metric differed from the original dataset metric (e.g., Fashion-MNIST and YFCC, which are natively Euclidean but evaluated under cosine similarity), ground truth was recomputed at runtime to ensure consistency.

\subsection{Additional plots and results}

\smallskip

\noindent
\textbf{Results on GloVe-200 and NYTimes.} 
We show the results for $R=32$ in the main paper, \Cref{fig:main-results}. The plots for $R=64$ are in \Cref{fig: glove R64} and \Cref{fig: nyt R64}, where Steiner augmentation still improves over the baseline except for Recall@100 on GloVe, with smaller margins. For NYTimes, similar to $R=32$, Steiner points help for Recall@100 as well.

\smallskip

\noindent
\textbf{Results on SIFT.}
On SIFT (\Cref{fig: sift R32 cosine}, \Cref{fig: sift R64 cosine}, \Cref{fig: sift R32 eucl}, and \Cref{fig: sift R64 eucl}), the baseline Vamana graph already achieves near-perfect recall (above $0.99$ for Recall@1 and Recall@10) at low beam widths. All three Steiner configurations slightly underperform the baseline, with the gap narrowing as the Steiner percentage increases ($25\% > 15\% > 5\%$). The effect is small in absolute terms (less than $1$ percentage point) but consistent across beam widths and degree bounds. Understanding precisely why Steiner augmentation does not yield improvements in this high-recall regime is an interesting direction for further investigation.

\smallskip

\noindent
\textbf{Results on Fashion-MNIST.}
On Fashion-MNIST in \Cref{fig: fmnist R32} and \Cref{fig: fmnist R64}, the baseline already exceeds Recall@1 of $0.95$ at the smallest beam width tested. All Steiner configurations closely track the baseline for Recall@1 and Recall@10, with the $25\%$ configuration showing mild degradation in Recall@100. As with SIFT, the baseline recall is sufficiently high that there is limited room for improvement.

\smallskip

\noindent
\textbf{Results on YFCC.}
YFCC in \Cref{fig: yfcc R32} and \Cref{fig: yfcc R64} exhibits the most pronounced gap between the Steiner-augmented and baseline graphs. The baseline reaches Recall@1 above $0.97$ at low beam widths, while all Steiner configurations fall short at matched computation budgets, with differences of $1$--$3$ percentage points. Among the three configurations, higher Steiner percentages perform closer to the baseline ($25\% > 15\% > 5\%$). The cause of the degradation on YFCC, which is notably more severe than on SIFT or Fashion-MNIST despite similarly high baseline recall, remains an open question.

\begin{figure}[H]
    \centering
    \includegraphics[width=.95\textwidth]{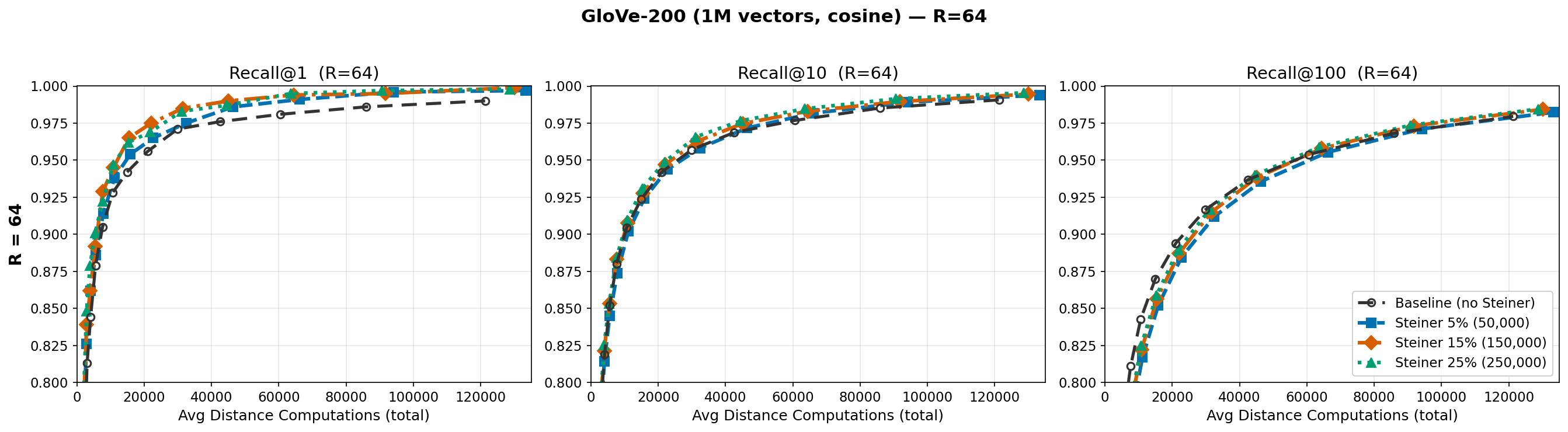}
    \caption{GloVe-200 (1M, cosine) — $R=64$}\label{fig: glove R64}
\end{figure}
\vspace{-.5em}
\begin{figure}[H]
    \centering
    \includegraphics[width=.95\textwidth]{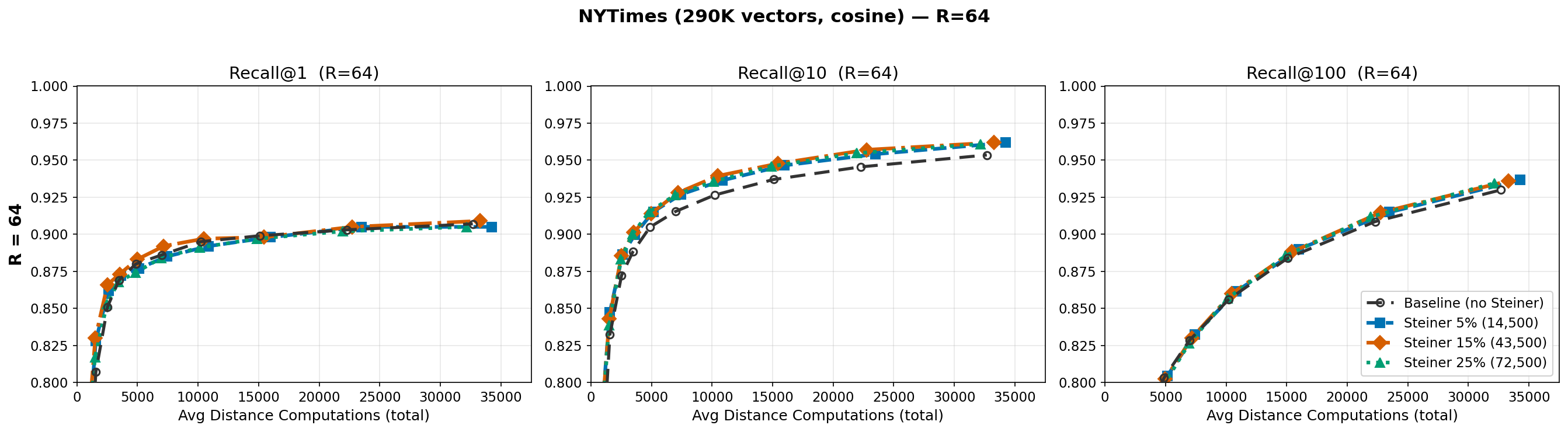}
    \caption{NYTimes (290K, cosine) — $R=64$}\label{fig: nyt R64}
\end{figure}
\vspace{-.5em}
\begin{figure}[H]
    \centering
    \includegraphics[width=.95\textwidth]{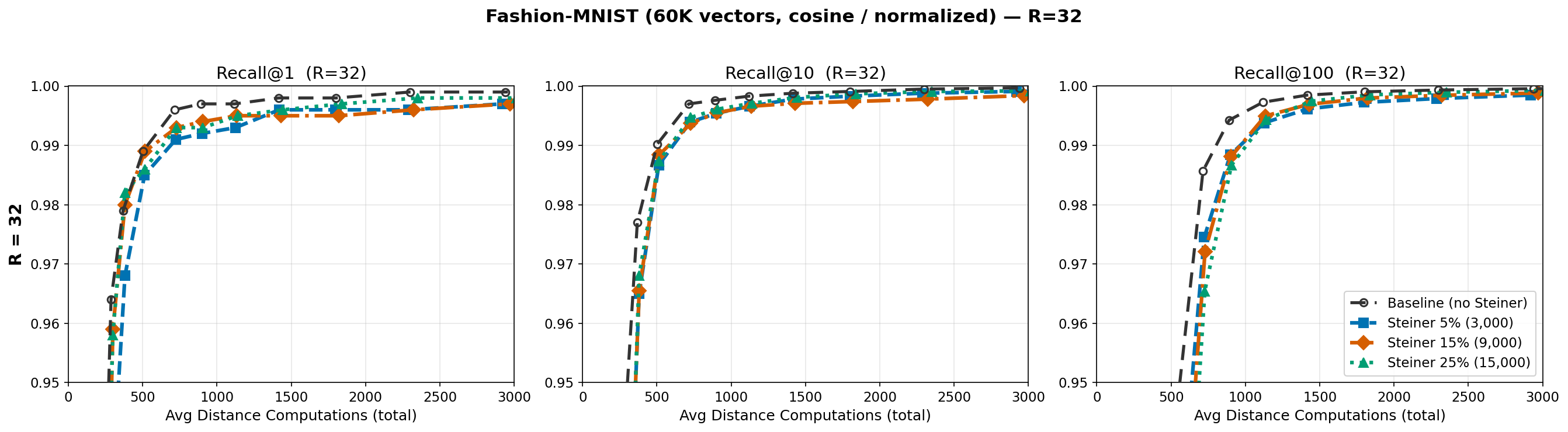}
    \caption{Fashion-MNIST (60K, cosine) — $R=32$}\label{fig: fmnist R32}
\end{figure}
\vspace{-.5em}
\begin{figure}[H]
    \centering
    \includegraphics[width=.95\textwidth]{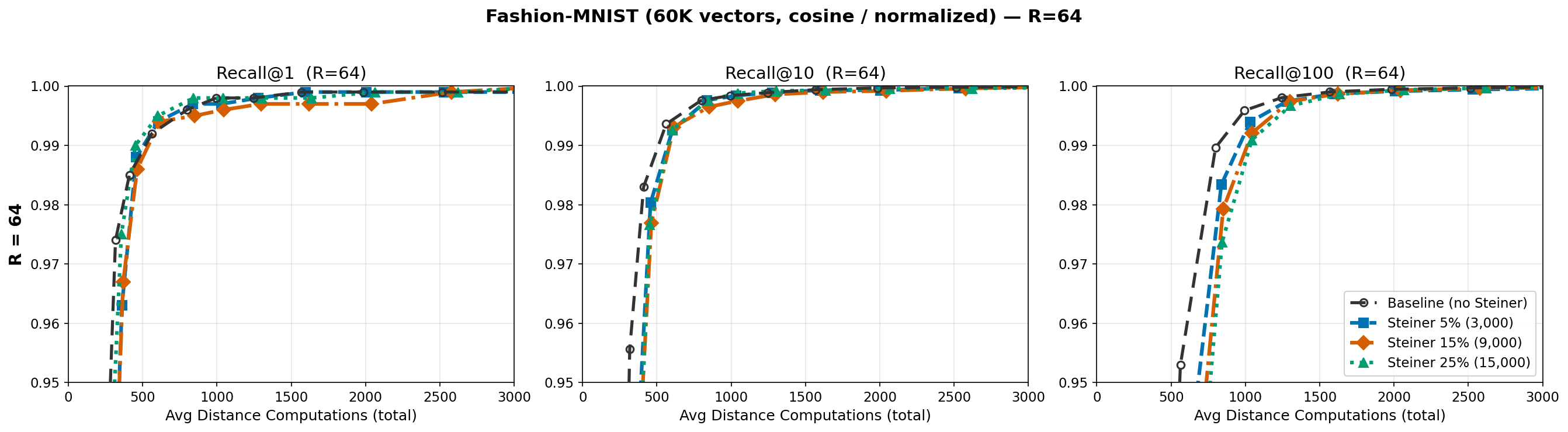}
    \caption{Fashion-MNIST (60K, cosine) — $R=64$}\label{fig: fmnist R64}
\end{figure}
\vspace{-.5em}
\begin{figure}[H]
    \centering
    \includegraphics[width=.95\textwidth]{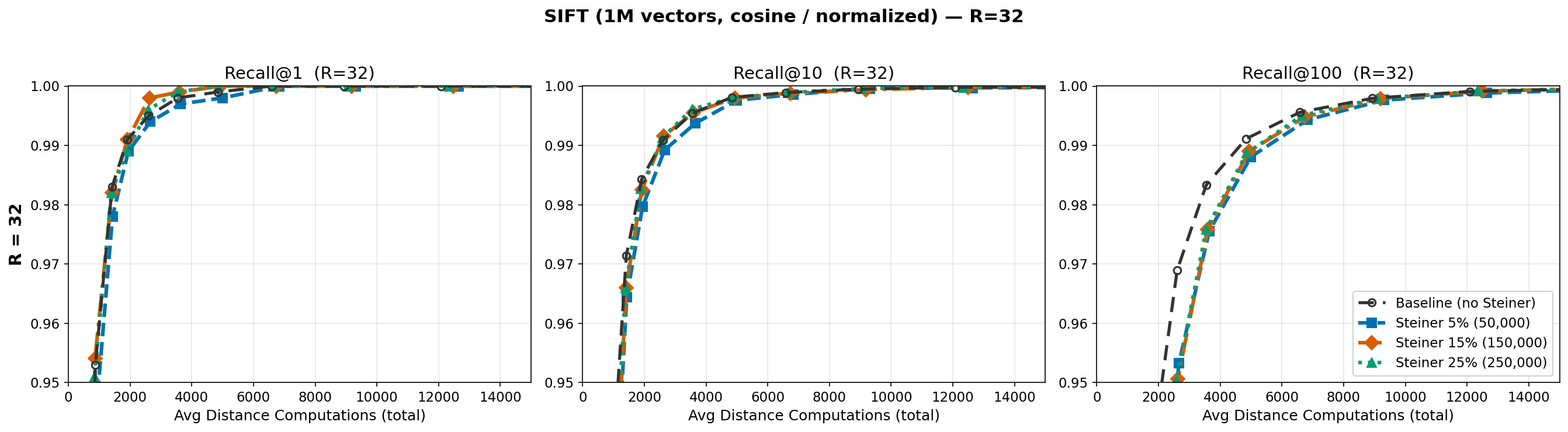}
    \caption{SIFT (1M, cosine) — $R=32$}\label{fig: sift R32 cosine}
\end{figure}
\vspace{-.5em}
\begin{figure}[H]
    \centering
    \includegraphics[width=.95\textwidth]{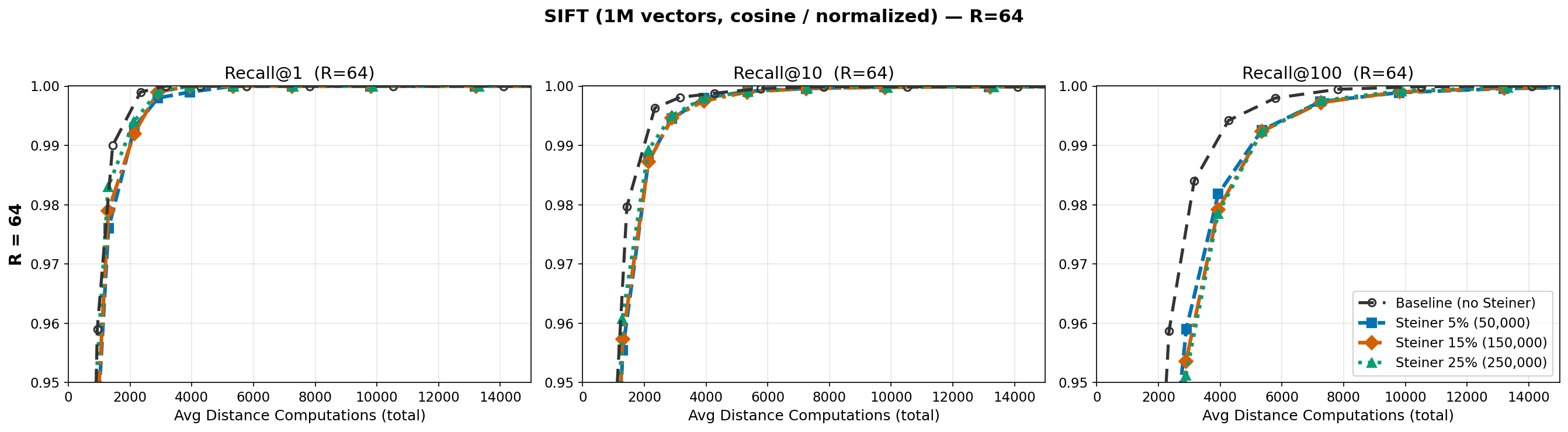}
    \caption{SIFT (1M, cosine) — $R=64$}\label{fig: sift R64 cosine}
\end{figure}
\vspace{-.5em}
\begin{figure}[H]
    \centering
    \includegraphics[width=.95\textwidth]{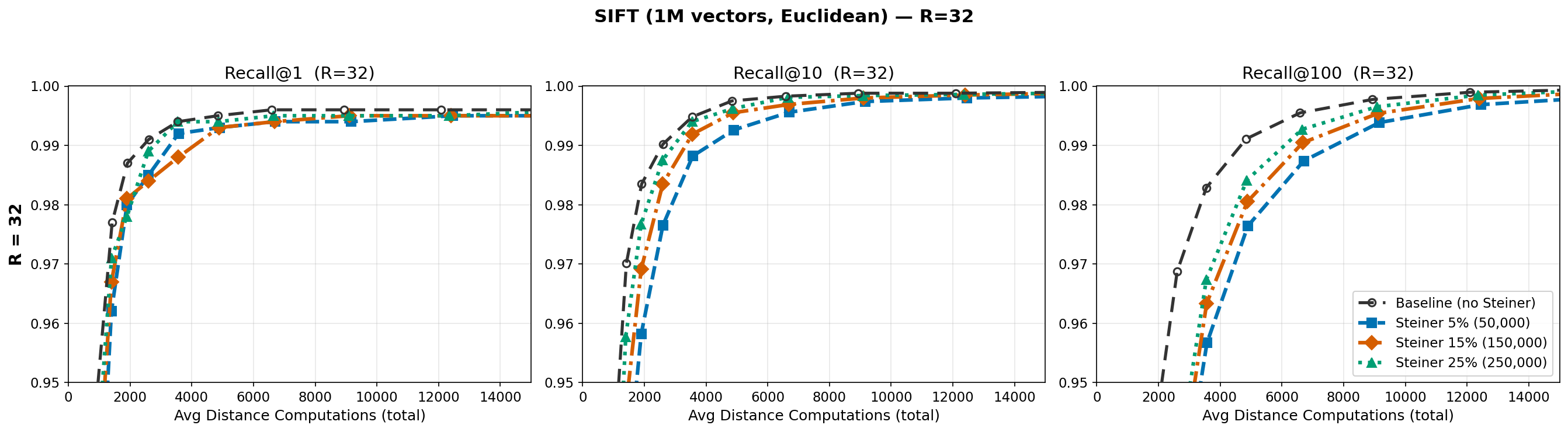}
    \caption{SIFT (1M, Euclidean) — $R=32$}\label{fig: sift R32 eucl}
\end{figure}
\vspace{-.5em}
\begin{figure}[H]
    \centering
    \includegraphics[width=.95\textwidth]{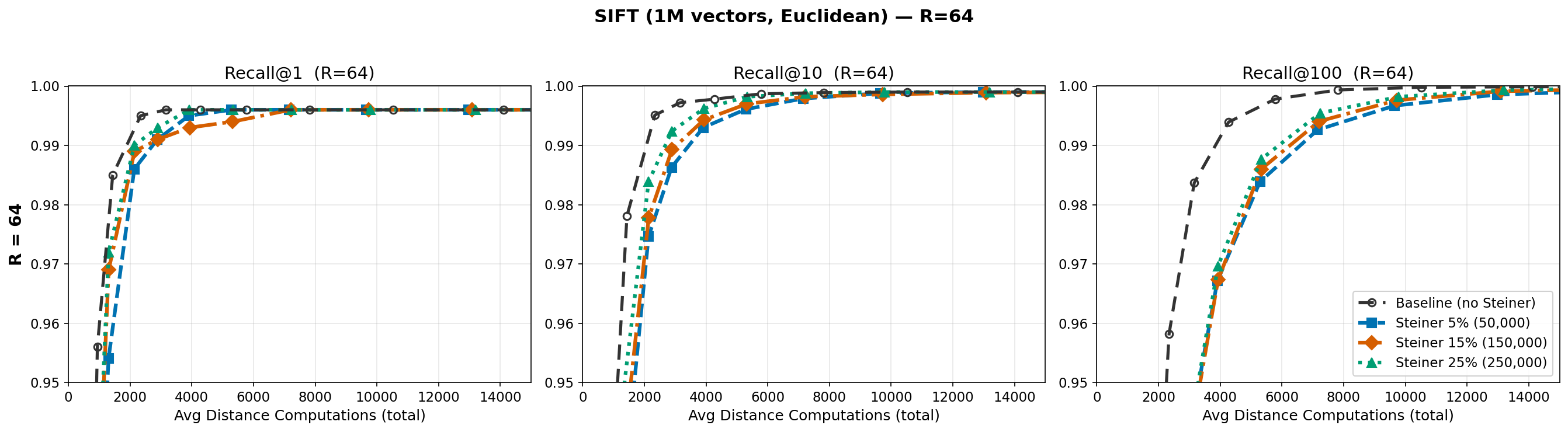}
    \caption{SIFT (1M, Euclidean) — $R=64$}\label{fig: sift R64 eucl}
\end{figure}
\vspace{-.5em}
\begin{figure}[H]
    \centering
    \includegraphics[width=.95\textwidth]{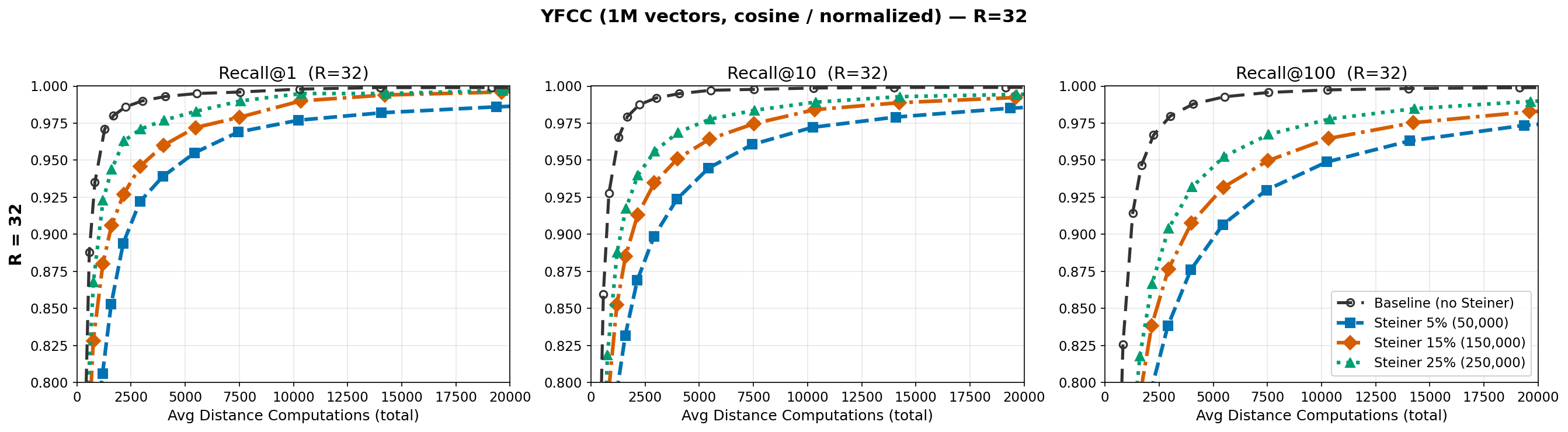}
    \caption{YFCC (1M, cosine) — $R=32$}\label{fig: yfcc R32}
\end{figure}
\vspace{-.5em}
\begin{figure}[H]
    \centering
    \includegraphics[width=.95\textwidth]{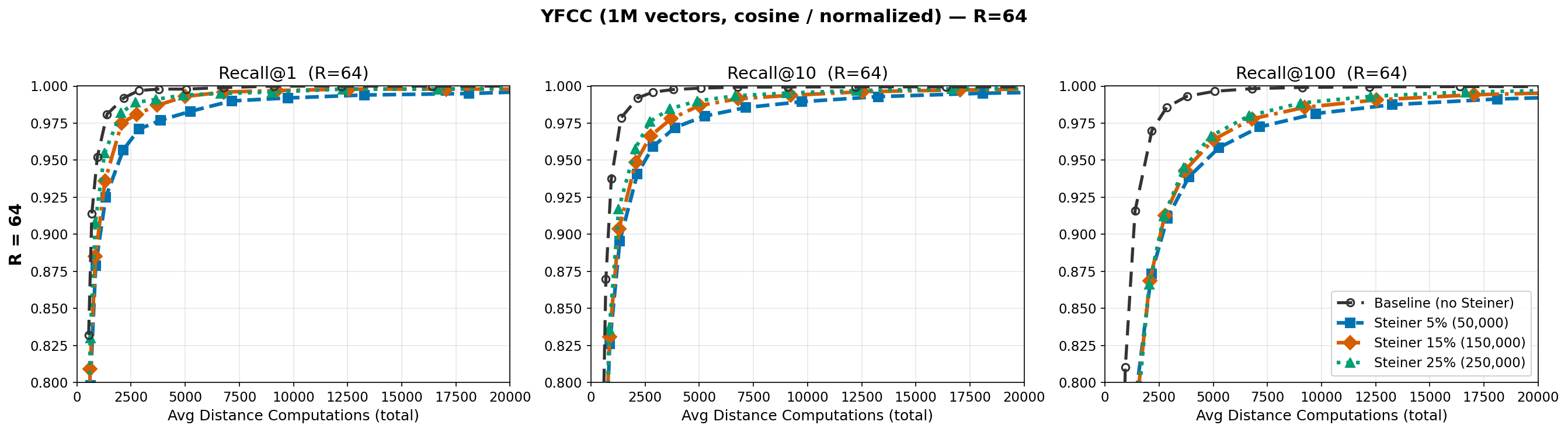}
    \caption{YFCC (1M, cosine) — $R=64$}\label{fig: yfcc R64}
\end{figure}

\end{document}